%% file: sigconf_unicorn_full.tex
\title{The Next 700 Policy Miners: A Universal Method for Building Policy Miners}
\author{
        Carlos Cotrini, Luca Corinzia, Thilo Weghorn, David Basin \\
                Department of Computer Science\\
        ETH Z\"urich\\
        \{ccarlos, luca.corinzia, thilo.weghorn, basin\}@inf.ethz.ch
}
\date{\today}

\documentclass[9pt]{extarticle}

\usepackage{extsizes}
\usepackage[a4paper, left=2.5cm, right=2.5cm, top=2cm, bottom=3cm]{geometry}
\usepackage{color, soul}
\usepackage{array}
\usepackage{amsmath}
\usepackage{amsthm}
\usepackage{amssymb}
\usepackage[linesnumbered, ruled]{algorithm2e}
\usepackage{booktabs}
\usepackage{caption}
\usepackage{enumitem}
\usepackage{hyperref}
\usepackage{pgfplots}
\usepackage{rotating}
\usepackage{subcaption}
\usepackage{stmaryrd}
\usepackage{libertine}
\usepackage{xcolor}

\theoremstyle{definition}
\newtheorem{ndefinition}{Definition}

\newtheorem{corollary}{Corollary}
\newtheorem{lemma}{Lemma}
\newtheorem{observation}{Observation}
\newtheorem{theorem}{Theorem}
\newtheorem{nexample}{Example}

\newtheorem{innercustomthm}{Lemma}
\newenvironment{customthm}[1]
  {\renewcommand\theinnercustomthm{#1}\innercustomthm}
  {\endinnercustomthm}

\newcommand\xqed[1]{%
  \leavevmode\unskip\penalty9999 \hbox{}\nobreak\hfill
  \quad\hbox{#1}}
\newenvironment{definition}{\begin{ndefinition}}{\xqed{$\square$}\end{ndefinition}}
\newenvironment{example}{\begin{nexample}}{\xqed{$\square$}\end{nexample}}

\newcommand{\polymver}{\textsc{Unicorn}}

\newcommand{\ricons}[1]{\underline{\mathsf{#1}}}

\newcommand{\Users}{\textbf{USERS}}
\newcommand{\Perms}{\textbf{PERMS}}

\newcommand{\NumRoles}{\mathit{N}}

\newcommand{\prob}[1]{\mathit{P}\left(#1\right)}

\newcommand{\Lang}{\mathcal{L}}

\newcommand{\carlos}[1]{
}

\newcommand{\Auth}{\mathit{Auth}}
\newcommand{\Roles}{\SortRoles}

\renewcommand{\prob}{\mathbb{P}}
\newcommand{\SortUsers}{\sortFont{USERS}}
\newcommand{\SortRoles}{\sortFont{ROLES}}
\newcommand{\SortPermissions}{\sortFont{PERMS}}
\newcommand{\sortFont}[1]{\textbf{#1}}

\newcommand{\lossfunc}{L}
\newcommand{\cons}[1]{\textsf{#1}}
\newcommand{\domi}[1]{\textsf{#1}}
\newcommand{\vari}[1]{\textit{#1}}
\newcommand{\Integers}{\textbf{INTS}}
\newcommand{\Strings}{\textbf{STRS}}
\newcommand{\Bool}{\textbf{BOOL}}

\begin{document}
\maketitle

\begin{abstract}
A myriad of access control policy languages have been and continue to be proposed. The design of policy miners for each such language is a challenging task that has required specialized machine learning and combinatorial algorithms. We present an alternative method, \emph{universal access control policy mining} ({\polymver}). We show how this method streamlines the design of policy miners for a wide variety of policy languages including ABAC, RBAC, RBAC with user-attribute constraints, RBAC with spatio-temporal constraints, and an expressive fragment of XACML. For the latter two, there were no known policy miners until now. 

To design a policy miner using {\polymver}, one needs a policy language and a metric quantifying how well a policy fits an assignment of permissions to users. From these, one builds the policy miner as a search algorithm that computes a policy that best fits the given permission assignment. We experimentally evaluate the policy miners built with {\polymver} on logs from Amazon and access control matrices from other companies. Despite the genericity of our method, our policy miners are competitive with and sometimes even better than specialized state-of-the-art policy miners. The true positive rates of policies we mined differ by only 5\% from the policies mined by the state of the art and the false positive rates are always below 5\%. In the case of ABAC, it even outperforms the state of the art.\looseness=-1
\end{abstract}

\input{unicorn_contents}

\bibliographystyle{plain}
\bibliography{unpolymerbiblio}

\appendix

\input{unicorn_appendix}

\end{document}

%% file: unicorn_contents.tex
\section{Introduction}
\label{sec:introduction}

\subsection{Motivation and research problem}
\label{sub:intro_problem}

Numerous access control policy languages have been proposed over the last decades, e.g., RBAC (Role-Based Access Control)~\cite{ferraiolo2001proposed}, ABAC (Attribute-Based Access Control)~\cite{hu2013guide}, XACML (eXtended Access-Control Markup Language)~\cite{godik2002oasis}, and new proposals are continually being developed, e.g., \cite{yang2016time,mukherjee2017attribute,
biswas2016label,cheng2016extended,
bhatt2017abac}. To facilitate the policy specification and maintenance process, policy miners have been proposed, e.g., \cite{frank2009probabilistic,xu2015mining,
mitra2013toward,bui19mining,gautam2017poster, cotrini2018rhapsody,molloy2012generative,amazonchallenge}.
These are algorithms that receive an assignment of permissions to users and output a policy that grants permissions to users that match as closely as possible the given assignment.\looseness=-1

Designing a policy miner is challenging and requires sophisticated combinatorial or machine-learning techniques. Moreover, policy miners are tailor-made for the specific policy language they were designed for and they are inflexible in that any modification to the miner's requirements necessitates its redesign and reimplementation. For example, miners that mine RBAC policies from access control matrices~\cite{frank2009probabilistic} are substantially different from those that mine RBAC policies from access logs~\cite{molloy2012generative}. As evidence for the difficulty of this task, despite extensive work in policy mining, no miner exists for XACML~\cite{godik2002oasis}, which is a well-known, standardized language.\looseness=-1 

Any organization that wishes to benefit from policy mining faces the challenge of designing a policy miner that fits its own policy language and its own requirements. This problem, which we examine in Section~\ref{sec:limitations}, is summarized with the following question: \emph{is there a more general and more practical method to design policy miners?}\looseness=-1

\begin{figure*}[h!btp]
\centering
\includegraphics[width=0.7\textwidth]{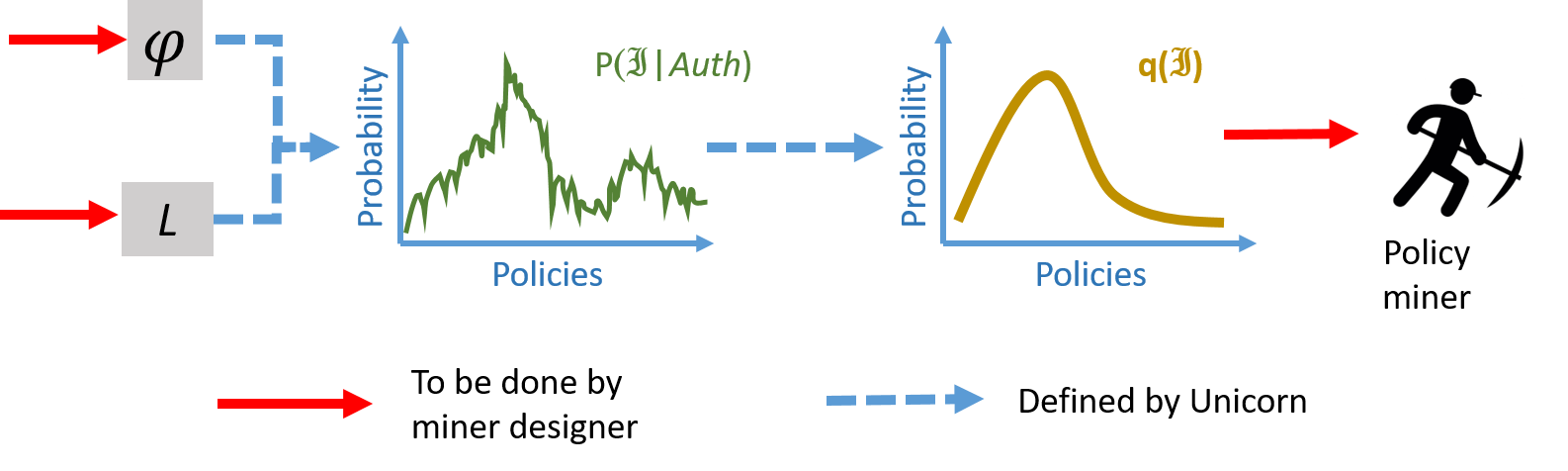}
\caption{Workflow for designing a policy miner using {\polymver}.}
\label{fig:unicorn}
\end{figure*}

\subsection{Contribution}
\label{sub:intro_contribution}

We propose a radical shift in the way policy miners are built. Rather than designing specialized mining algorithms, one per policy language, we propose {\polymver}, \emph{a universal method for building policy miners}. Using this method, the designers of policy miners no longer must be experts in machine learning or combinatorial optimization to design effective policy miners. Our method gives a step-by-step procedure to build a policy miner from just the \emph{policy language} and an \emph{objective function} that measures how well a policy fits an assignment of permissions to users.

Let $\Gamma$ be a policy language. We sketch below and in Figure~\ref{fig:unicorn} the workflow for designing a policy miner for $\Gamma$ using {\polymver}.

\paragraph{Policy language and objective function (Sections~\ref{sec:uni_syntax} and~\ref{sec:prob_dist})} The miner designer specifies a \emph{template formula for $\Gamma$} in a fragment $\Lang$ of first-order logic. Template formulas are explained in Section~\ref{sec:uni_syntax}. 
The designer also specifies an objective function $L$ that measures how well a policy fits a permission assignment. 
\vspace{-5pt}
\paragraph{Probability distribution (Section~\ref{sec:prob_dist})} From $\varphi$ and $L$, we define a probability distribution $\mathbb{P}$ on policies, conditioned on permission assignments. A permission assignment is a relation between the set of users and the set of permissions. 
The policy miner is a program that receives as input a permission assignment $\mathit{Auth}$ and aims to compute the most likely policy conditioned on $\Auth$; that is, the policy $\mathfrak{I}$ that maximizes $\mathbb{P}\left(\mathfrak{I} \mid \mathit{Auth}\right)$. 
\vspace{-5pt}
\paragraph{Approximation (Section~\ref{sec:build_var_infer})} Computing $\max_{\mathfrak{I}}\mathbb{P}\left(\mathfrak{I} \mid \mathit{Auth}\right)$ takes time exponential in the size of $\mathit{Auth}$ and $\mathfrak{I}$, encoded as strings. Moreover, the function $\mathbb{P}\left(\mathfrak{I} \mid \mathit{Auth}\right)$ has many local maxima. Hence, we use \emph{deterministic annealing} and \emph{mean-field approximation}~\cite{rose1992vector,rose1998deterministic, blei2017variational, bishop2006pattern} to derive an iterative procedure that computes a distribution $q$ on policies that approximates $\mathbb{P}\left(\mathfrak{I} \mid \mathit{Auth}\right)$. Computing $\arg\max_{\mathfrak{I}}q\left(\mathfrak{I}\right)$ takes time polynomial in the size of $\mathit{Auth}$ and $\mathfrak{I}$.
\vspace{-5pt}
\paragraph{Implementation (Section~\ref{sec:polymer})} The policy miner is a procedure that computes and maximizes $q$. One need not understand mean-field approximations or deterministic annealing to implement the policy miner. We provide a set of rewriting rules and pseudocode that guide step by step $q$'s computation and maximization (see Algorithm~\ref{algo:varinfer} and Lemma~\ref{lem:expectations_flas}).\looseness=-1

In summary, designing a policy miner for a policy language previously required expertise in machine learning and combinatorial algorithms. {\polymver} reduces this to the task of specifying a template formula and implementing $q$'s maximization. We illustrate how specifying template formulas requires only the background in first-order logic provided in this paper and how it amounts to just formalizing the language's semantics in first order logic, a task that is substantially simpler than designing a machine-learning or a combinatorial algorithm.

\subsection{Applications and evaluation}
\label{sub:intro_using}

Using {\polymver}, we have built miners for different policy languages like RBAC, ABAC, and RBAC with user attributes. Furthermore, we have built policy miners for RBAC with spatio-temporal constraints and an expressive fragment of XACML, for which no miner existed before. We present them in Sections~\ref{sec:examples} and ~\ref{sub:st_rbac} and in the appendix.\looseness=-1

In Section~\ref{sec:experiments}, we conduct an extensive experimental evaluation using datasets from all publicly available real-world case studies on policy mining. We compare the miners we built with state-of-the-art miners on both real-world and synthetic datasets. The true positive rates of the policies mined by our miners are within 5\% of the true positive rates of the policies mined by the state of the art. For policy languages like XACML or RBAC with spatio-temporal constraints, the true positive rates are above 75\% in all cases and above 80\% in most of them. The false positive rates are always below 5\%. For ABAC policies, we mine policies with a substantially lower complexity and higher precision than those mined by the state of the art. This demonstrates that with {\polymver} we can build a wide variety of policy miners, including new ones, that are competitive with or even better than the state of the art.\looseness=-1

{\polymver}'s effectiveness follows from the wide applicability of deterministic annealing (DA). This technique has been applied to different optimization problems like the traveling salesman problem~\cite{rose1998deterministic}, clustering~\cite{rose1992vector}, and image segmentation~\cite{hofmann1997pairwise}. DA can also be applied to policy mining. However, in our case, computing the distribution $\mathbb{P}\left(\mathfrak{I} \mid \mathit{Auth}\right)$ required by DA is intractable. Hence, we use mean-field approximation (MFA) to compute a distribution $q$ that approximates $\mathbb{P}$. This distribution $q$ is much easier to compute. Moreover, our approach of DA with MFA turns out to generalize to a wide variety of policy languages.

We examine related work and draw conclusions in Sections~\ref{sec:related_work} and~\ref{sec:uni_conclusion}. For details on deterministic annealing and mean-field approximation, we refer to the literature ~\cite{rose1992vector, rose1998deterministic, blei2017variational, bishop2006pattern}.\looseness=-1 

\vspace{-5pt}

\section{Preliminaries}

\subsection{Policy mining}
\label{sub:policy_miners}

Organizations define \emph{organizational policies} that specify which permissions each user in the organization has. Such policies are usually described in a high-level language. To be machine enforceable, policy administrators must specify this policy as an \emph{(access control) policy} in a machine-readable format. This policy assigns \emph{permissions} to \emph{users} and is formalized in a \emph{policy language}. The policy is then enforced by mechanisms that intercept each \emph{request} (a pair consisting of a user and a permission) and check whether it is authorized.\looseness=-1 

Organizations are highly dynamic. New users come and existing users may go. Moreover, groups of users may be transferred to other organizational units. Such changes induce changes in the access control policy, which are usually manually implemented, giving rise to the following problems. First, the policy may become convoluted and policy administrators no longer have an overview on who is authorized to do what. Second, policy administrators may have granted to users more permissions than needed to do their jobs. This makes the organization vulnerable to abuse by its own users, who may exploit the additional permissions and harm the organization.\looseness=-1

To address these problems, numerous \emph{policy miners} have been proposed~\cite{frank2009probabilistic,xu2015mining,
mitra2013toward,bui19mining,gautam2017poster, cotrini2018rhapsody,molloy2012generative,amazonchallenge}. We describe some of them in Section~\ref{sub:status_quo}. Policy miners are algorithms that receive as input the current \emph{permission assignment}, which is a relation between the set of users and the set of permissions. The permission assignment might be given as an access control matrix or a log of access requests showing the access decisions previously made for each request. It describes the organization's implemented knowledge on which permissions should be assigned to which users. The miner then constructs a policy that is as consistent as possible with the permission assignment and can be expressed using the organization's policy language.

A policy miner aims to solve the two problems mentioned above. First, it can mine succinct policies that grant permissions consistent with the given permission assignment. 
Second, policy miners can mine policies that assign only those permissions that users necessarily need. An administrator can then compare the mined policy with the currently implemented policy in order to detect permissions that are granted by the current policy, but that are not being exercised by the users. Policy administrators can then inspect those permissions and decide if they are necessary for those users.

The problem of policy mining is defined as follows. Given a permission assignment and an objective function, compute a policy that minimizes the objective function. Usually, objective functions measure how well a policy fits a permission assignment and how complex a policy is. We give examples of objective functions later in Sections~\ref{sec:prob_dist} and~\ref{sec:examples}.

\subsection{Quality criteria for policy miners}
\label{sub:quality_criteria}

Policy miners can be regarded as machine-learning algorithms. Therefore, they are evaluated by the quality of the policies they mine, and here two criteria are used:

\paragraph{Generalization~\cite{frank2013role,cotrini2018rhapsody,molloy2012generative}} A mined policy should not only authorize requests consistent with the given permission assignment. It must also correctly decide what \emph{other} permissions should be granted to users who perform similar functions in the organization. This is particularly important when mining from logs. For example, if most of students in a university have requested and been granted access to a computer room, then the mined policy should grant all students access to the computer room rather than just to those who previously requested access to it. For a formal definition of generalization, we refer to previous work and standard references in machine learning~\cite{cotrini2018rhapsody, frank2013role, bishop2006pattern}. One popular machine-learning method to evaluate generalization is cross-validation~\cite{friedman2001elements, bishop2006pattern}, which we recall in Appendix~\ref{sub:cross_val}.

\paragraph{Complexity~\cite{bui19mining, xu2015mining}} A mined policy should not be unnecessarily complex, as the policies are usually reviewed and audited by humans. This is especially important when mining with the goal of refactoring an existing policy or migrating to a new policy language. However, there is no standard formalization of a policy's complexity, not even for established policy languages like RBAC or ABAC. Each previous work has defined its own metrics to quantify complexity~\cite{cotrini2018rhapsody, xu2014mining, xu2015mining, frank2010definition}. We discuss some of these metrics in Section~\ref{sec:examples} and show how {\polymver} is able to work with all of them.

\section{The problem of designing policy miners}
\label{sec:limitations}

\subsection{Status quo: specialized solutions}
\label{sub:status_quo}

Numerous policy languages exist for specifying access control policies, which fulfill different organizational requirements. 
Moreover, new languages are continually being proposed. Some of them formulate new concepts, like extensions of RBAC that can express temporal and spatial constraints~\cite{ray2007spatio, chen2008spatio, kumar2006strbac, toahchoodee2009ensuring, chandran2005lot, aich2007starbac, cui2007ex}. Other languages facilitate policy specification in specialized settings such as distributed systems~\cite{godik2002oasis, tsankov2014decentralized} or social networks~\cite{fong2011relationship}.\looseness=-1

Motivated by the practical problem of maintaining access control configurations, researchers have proposed policy miners for a variety of policy languages. Moreover, for some policy languages, these miners optimize different objectives. For example, initial RBAC miners mined policies with a minimal number of roles~\cite{vaidya2007role, lu2008optimal, schlegelmilch2005role, vaidya2006roleminer, zhang2007role}. Subsequent miners mined policies that are as consistent as possible with the user-attribute information~\cite{molloy2012generative, frank2013role, xu2012algorithms}. \looseness=-1

The development of policy miners is non-trivial and generally requires sophisticated combinatorial and machine-learning algorithms. Recent ABAC miners have used association rule mining~\cite{cotrini2018rhapsody} and classification trees~\cite{chari2016generation}. The most effective RBAC miners use deterministic annealing~\cite{frank2013role} and latent Dirichlet allocation~\cite{molloy2012generative}.

The proposed miners are so specialized that it is usually unclear how to apply them to other policy languages or even to extensions of the languages for which they were conceived. For example, different extensions of RBAC that support spatio-temporal constraints have been proposed over the last two decades, e.g.,~\cite{ray2007spatio, chen2008spatio, kumar2006strbac, toahchoodee2009ensuring, chandran2005lot, aich2007starbac, cui2007ex}. However, not a single miner has been proposed for these extensions. Miners have only recently emerged that mine RBAC policies with constraints, albeit only temporal ones~\cite{mitra2015generalized, mitra2016mining, stoller2016mining}. As a result, if an organization wants to use a specialized policy language, it must invent its own policy miner, which is challenging and time-consuming.

\subsection{Alternative: A universal method}
\label{sub:contribution} 

To facilitate the development of policy miners, we propose a new method, \emph{universal access control policy mining} ({\polymver}). With this method, organizations no longer need to spend substantial effort designing specialized policy miners for their unique and specific policy languages; they only need to perform the following tasks (see also Figure~\ref{fig:unicorn}). First, they specify a \emph{template formula} $\varphi$ for the organization's policy language. We explain later in Section~\ref{sec:uni_syntax} what a template formula is. Second, they specify an objective function. Finally, 
they implement the miner as indicated by the algorithm template in Section~\ref{sec:polymer}.
We formalize these tasks in the next sections.\looseness=-1

\section{A universal policy language}
\label{sec:uni_syntax}
In order to obtain a universal method, we need a framework for specifying policy languages. We choose \emph{many-sorted first-order logic}~\cite{enderton2001mathematical, ebbinghaus2013mathematical}, which has been used to model and reason about numerous policy languages, e.g.~\cite{arkoudas2014sophisticated,jha2014security,cotrini2015analyzing,turkmen2015analysis}.

Let $\Gamma$ be a policy language for which we want to design a policy miner. In this section we explain the first task: the miner designer must specify a template formula $\varphi_{\Gamma}$ for $\Gamma$. This is a first-order formula that fulfills some conditions that we explain later in Definition~\ref{def:uni_template_formula}. We show how $\Gamma$ can be identified with $\varphi_{\Gamma} \in \Lang$ and how policies in $\Gamma$ can be identified with \emph{interpretation functions} that interpret $\varphi_{\Gamma}$'s symbols. We thereby reduce the problem of designing a policy miner to designing an algorithm that searches for a particular interpretation function.\looseness=-1

We start by recalling first-order logic (Section~\ref{sec:back_fol}). Then we provide some intuition on template formulas using RBAC (Section~\ref{sub:motivating_example}). Afterwards, we propose a fragment $\Lang$ of first-order logic that is powerful enough to contain template formulas for a variety of policy languages like RBAC, ABAC, and an expressive fragment of XACML (Section~\ref{sub:uni_lang_def}). We then define template formulas (Section~\ref{sub:uni_template_formulas}) and give an example of a template formula for RBAC (Section~\ref{sub:example_rbac}).\looseness=-1

\subsection{Background in first-order logic}
\label{sec:back_fol}

We provide here an overview of basic many-sorted first-order logic and conventions we employ. 
The reader familiar with logic can read this section lightly.
We work only with \emph{finite} first-order structures. That is, structures whose carrier sets are finite. Later, in our examples, we will see that finite structures are still powerful enough to model practical scenarios, as organizations do not need to handle infinite sets. Even for the case of strings and integers, organizations often only use a finite subset of them. \looseness=-1
%

\begin{definition}
A \emph{signature} is a tuple $\left(\mathbb{S}, \mathbb{R}, \mathbb{F}, \mathbb{V}\right)$ fulfilling the following, where $\mathbb{S}$ is a finite non-empty set of \emph{sorts}, $\mathbb{R}$ is a finite non-empty set of \emph{relation symbols}, $\mathbb{F}$ is a finite non-empty set of \emph{function symbols}, and $\mathbb{V}$ is a countable set of \emph{variables}.

Each relation and each function symbol has an associated \emph{type}, which is a sequence of sorts. 
Furthermore, we assume the existence of two sorts $\Users, \Perms \in \mathbb{S}$, denoting the users and the permissions in the organization, respectively. We also assume the existence of the sorts $\Bool$, $\Integers$, $\Strings$, which represent Boolean values, integers, and strings, respectively. 
\end{definition}

We denote sorts with \textbf{\textsf{CAPITAL BOLD}} letters, relation symbols with $\mathit{CAPITAL}$ $\mathit{ITALIC}$ letters, and function symbols and variables with $\mathit{small}$ $\mathit{italic}$ letters. To agree with standard notation, we write a relation symbol's type $(\sortFont{S}_1, \ldots, \sortFont{S}_k)$ as $\sortFont{S}_1 \times \ldots \times \sortFont{S}_k$ instead. We write a function's symbol's type $(\sortFont{S}_1, \ldots, \sortFont{S}_k)$ as $\sortFont{S}_1 \times \ldots \times \sortFont{S}_{k-1} \to \sortFont{S}_k$ instead. We allow $k = 1$ and, in that case, we call function symbols \emph{constant symbols}. We denote constant symbols with \textsf{small serif} letters.\looseness=-1

\begin{definition}
Let $\Sigma$ be a signature. We define \emph{(first-order) terms} as those expressions built from $\Sigma$'s variables and function symbols in the standard way. We also define \emph{(first-order) formulas} as those expressions obtained from terms by using relation symbols, terms, and logical operators in the standard way.
\label{def:bkg_formulas}
\end{definition}

We only allow well-typed terms and formulas and associate to every term a type in the standard way. In addition, we consider only quantifier-free formulas. For a formula $\varphi$, if $\{x_1, \ldots, x_n\}$ is the set of all variables occurring in it, then we sometimes write $\varphi\left(x_1, \ldots, x_n\right)$ instead of $\varphi$ to clarify which variables occur in $\varphi$.


\begin{definition}
Let $\Sigma$ be a signature. A \emph{$\Sigma$-structure} is a pair $\mathbb{K} = \left(\mathfrak{S}, \mathfrak{I}\right)$. Here, $\mathfrak{S}$ is a function mapping each sort $\sortFont{S}$ in $\Sigma$ to a finite non-empty set $\sortFont{S}^\mathfrak{S}$, called \emph{$\sortFont{S}$'s carrier set}. $\mathfrak{S}$ must map $\Bool$, $\Integers$, and $\Strings$ to the sets of Boolean values, a finite set of integers, and a finite set of strings, respectively. $\mathfrak{I}$ is a function mapping (i) each relation symbol $R$ in $\Sigma$ of type $\sortFont{S}_1 \times \ldots \times \sortFont{S}_k$ to a relation $R^\mathfrak{I} \subseteq \sortFont{S}^\mathfrak{S}_1 \times \ldots \times \sortFont{S}^\mathfrak{S}_k$ and (ii) each function symbol $f$ in $\Sigma$ of type $\sortFont{S}_1 \times \ldots \times \sortFont{S}_{k-1} \to \sortFont{S}_k$ to a function $f^\mathfrak{I} : \sortFont{S}^\mathfrak{S}_1 \times \ldots \times \sortFont{S}^\mathfrak{S}_{k-1} \to \sortFont{S}^\mathfrak{S}_k$. In particular, a constant symbol of sort $\sortFont{S}$ is mapped to an element in $\sortFont{S}^{\mathfrak{S}}.$
For any symbol $W$ in $\Sigma$, we call $W^{\mathfrak{I}}$, \emph{$\mathbb{K}$'s interpretation of $W$}. The function $\mathfrak{I}$ is called an \emph{interpretation function}.
\label{def:bkg_sigma_structure}
\end{definition}

When $\Sigma$ is irrelevant or clear from the context, we simply say structure instead of $\Sigma$-structure. We denote elements of carrier sets with \textsf{small serif} letters like $\domi{a}$ and $\domi{b}$. \looseness=-1

Let $\left(\mathfrak{S}, \mathfrak{I}\right)$. The interpretation function $\mathfrak{I}$ gives rise in the standard way to a function that maps any formula $\varphi\left(x_1, \ldots, x_n\right)$, with $x_i$ of sort $\sortFont{W}_i$, to a relation $\varphi^{\mathfrak{I}} \subseteq \sortFont{W}_1^{\mathfrak{S}} \times \ldots \times \sortFont{W}_n^{\mathfrak{S}}$. For $(\domi{a}_1, \ldots \domi{a}_n) \in \sortFont{W}_1^\mathfrak{S} \times \ldots \times \sortFont{W}_n^\mathfrak{S}$, $\varphi^{\mathfrak{I}}(\domi{a}_1, \ldots \domi{a}_n)$ holds if the formula $\varphi$ evaluates to true after replacing each $x_i$ with $a_i$. 

\subsection{Motivating example}
\label{sub:motivating_example}


We present an example of a template formula $\varphi^{\mathit{RBAC}}_N$ for the language $\Gamma_N$ of all RBAC policies with at most $N$ roles. We then show that \emph{every RBAC policy in $\Gamma$ can be identified with an interpretation function}. With this example, we provide some intuition on an argument we give later in Section~\ref{sub:uni_template_formulas}: mining a policy in a policy language $\Gamma$ is equivalent to searching for an interpretation function that interprets the symbols occurring in a template formula for $\Gamma$.

\begin{definition}
An \emph{RBAC policy} is a tuple $\pi = (U, \mathit{Ro}, P, \mathit{Ua}, \mathit{Pa})$. $U$ and $P$ are non-empty sets denoting, respectively, the sets of users and permissions in an organization. $\mathit{Ro}$ is a set denoting the organization's roles. $\mathit{Ua} \subseteq \mathit{U} \times \mathit{Ro}$ and $\mathit{Pa} \subseteq \mathit{Ro} \times \mathit{P}$ are binary relations. The policy $\pi$ \emph{assigns} a permission $\domi{p} \in P$ to a user $\domi{u} \in U$ if $(\domi{u},\domi{p}) \in \mathit{Ua} \circ \mathit{Pa}$ (i.e., if there is a role $\domi{r} \in \mathit{Ro}$ such that $(\domi{u},\domi{r}) \in \mathit{Ua}$ and $(\domi{r},\domi{p}) \in \mathit{Pa}$).\looseness=-1
\label{def:pre_rbac}
\end{definition}

Consider the language $\Gamma_N$ of RBAC policies with at most $N$ roles. We now present a template formula for $\Gamma_N$. We only provide some intuition here and give formal justifications in Section~\ref{sub:example_rbac}. Let $\Sigma$ be a signature with two relation symbols $\mathit{UA}$ and $\mathit{PA}$ of types $\Users \times \Roles$ and $\Roles \times \Perms$, respectively. Let
\begin{equation}
\varphi^{\mathit{RBAC}}_N\left(\vari{u}, \vari{p}\right) := \bigvee_{i \leq N} \left(\mathit{UA}\left(\vari{u}, \cons{r}_i\right) \land \mathit{PA}\left(\cons{r}_i, \vari{p}\right)\right).
\label{eq:example_forbac_formula}
\end{equation}
Here, $\vari{u}$ and $\vari{p}$ are variables of sorts $\SortUsers$ and $\SortPermissions$, respectively, and $\cons{r}_i$, for $i \leq N$, is a constant of sort $\SortRoles$. We now make two observations about $\varphi^{\mathit{RBAC}}_N\left(\vari{u}, \vari{p}\right)$.

\paragraph{1) Each RBAC policy in $\Gamma_N$ corresponds to at least one interpretation function} Note that for any $\Sigma$-structure $\mathbb{K} = \left(\mathfrak{S}, \mathfrak{I}\right),$ the tuple 
\begin{equation}
\pi_{\mathbb{K}} = \left(\Users^{\mathfrak{S}}, \Roles^{\mathfrak{S}}, \Perms^{\mathfrak{S}}, \mathit{UA}^{\mathfrak{I}}, \mathit{PA}^{\mathfrak{I}}\right)
\end{equation}
is an RBAC policy. Conversely, one can show that every RBAC policy in $\Gamma_N$ can be associated with a $\Sigma$-structure. Observe now that, when an organization wants to mine an RBAC policy, $\mathfrak{S}$ is already known. Indeed, the organization knows the set of users and permissions. It may not known the set of roles, but it can deduce them from $\mathit{UA}^{\mathfrak{I}}$ and $\mathit{PA}^{\mathfrak{I}}$, once it knows $\mathfrak{I}$. Analogously, for all policy languages we studied, we observed that $\mathfrak{S}$ was always known by the organization. Therefore, we always assume $\mathfrak{S}$ given and fixed and we conclude that every RBAC policy in $\Gamma_N$ corresponds to at least one interpretation function. 

\paragraph{2) The formula $\varphi^{\mathit{RBAC}}_N$ describes $\Gamma_N$'s semantics} More precisely, if $\pi_{\mathbb{K}}$ has at most $N$ roles, then for any user $\domi{u}$ in $\mathbb{K}$ and any $\domi{p}$ in $\mathbb{K}$: $\pi_{\mathbb{K}}$ assigns $\domi{p}$ to $\domi{u}$ iff $\left(\varphi^{\mathit{RBAC}}_N\right)^{\mathfrak{I}}\left(\domi{u}, \domi{p}\right)$. This follows from two arguments. First, by definition, $\pi_{\mathbb{K}}$ assigns $\domi{p}$ to $\domi{u}$ if $(\domi{u}, \domi{p}) \in \mathit{UA}^{\mathfrak{I}} \circ \mathit{PA}^{\mathfrak{I}}$. Second, $\mathit{UA}^{\mathfrak{I}} \circ \mathit{PA}^{\mathfrak{I}} = \left(\varphi^{\mathit{RBAC}}_N\right)^{\mathfrak{I}}$.

These two observations describe the essence of a template formula. Template formulas define (i) how interpretation functions can represent policies of a policy language and (ii) how a policy (represented by an interpretation function) decides if a permission is assigned to a user.


\subsection{Language definition}
\label{sub:uni_lang_def}

Template formulas are built from the fragment $\Lang$ of quantifier-free first-order formulas.

For any signature, we require the organization to specify, for every relation and function symbol, whether it is \emph{rigid} or \emph{flexible}. Rigid symbols are those for which the organization already knows the interpretation function. Flexible symbols are those for which an interpretation function must be found using mining. For example, a function that maps each user to a unique identifier should be modeled with a rigid function symbol, as the organization is not interested in mining new identifiers. In contrast, when mining RBAC policies, one should define a flexible relation symbol to denote the assignment of roles to users, as the organization does not know this assignment and wants to compute it using mining.

\newcommand{\fixedinter}{\mathfrak{I}_r}
%

Let $\mathbb{K} = \left(\mathfrak{S}, \mathfrak{I}\right)$ be a structure. We can see $\mathfrak{I}$ as the union of two interpretation functions $\fixedinter$ and $\mathfrak{I}_f$, where $\fixedinter$ takes as input rigid symbols and $\mathfrak{I}_f$ takes as input flexible symbols. The goal of policy mining is to search for an interpretation function $\mathfrak{I}_f$ for the flexible symbols that minimizes an objective function. It does not need to search for $\mathfrak{S}$ as these function defines the carrier sets for sorts like $\Users$ and $\Perms$, which the organization already knows. It does not need to search for $\fixedinter$ either. Hence, we assume that $\mathfrak{S}$ and $\fixedinter$ are fixed and known to the organization. We also let $U = \Users^\mathfrak{S}$ and $P = \Perms^\mathfrak{S}$. We \underline{underline} rigid symbols and do not distinguish between $\underline{W}$ and $\underline{W}^{\mathfrak{I}_r}$.\looseness=-1


\subsection{Template formulas}
\label{sub:uni_template_formulas}

We now formalize template formulas. Let $\Gamma$ be a policy language and let $\mathit{Pol}(\Gamma)$ be the set of all policies that can be specified with $\Gamma$. Suppose also that the set of access requests is modeled with a set $\sortFont{T}^\mathfrak{S}_1 \times \ldots \times \sortFont{T}^\mathfrak{S}_\ell$, where $\sortFont{T}_1, \ldots, \sortFont{T}_\ell$ are sorts. For example, for RBAC and many other policy languages that we discuss here, the set of requests is $U \times P = \SortUsers^{\mathfrak{S}} \times \SortPermissions^{\mathfrak{S}}$. 

We assume that the semantics of $\Gamma$ defines a relation $\mathit{assign}_{\Gamma} \subseteq \mathit{Pol}\left(\Gamma\right) \times \sortFont{T}^\mathfrak{S}_1 \times \ldots \times \sortFont{T}^\mathfrak{S}_\ell$, such that for $\left(\domi{t}_1, \ldots, \domi{t}_\ell\right) \in \sortFont{T}^\mathfrak{S}_1 \times \ldots \times \sortFont{T}^\mathfrak{S}_\ell$ and $\pi \in \mathit{Pol}(\Gamma)$, $(\pi, \domi{t}_1, \ldots, \domi{t}_\ell) \in \mathit{assign}_{\Gamma}$ iff $\pi$ authorizes $\left(\domi{t}_1, \ldots, \domi{t}_\ell\right)$. For example, in RBAC, $\left(\pi, \domi{u}, \domi{p}\right) \in \mathit{assign}_{\mathit{RBAC}}$ iff $\pi$ assigns $\domi{p}$ to $\domi{u}$.

\begin{definition}
Let $\Gamma$ be a policy language and $\varphi(t_1, \ldots, t_\ell)$ be a formula in $\Lang$, where $t_1, \ldots, t_\ell$ are variables of sorts $\sortFont{T}_1, \ldots, \sortFont{T}_\ell$, respectively. The formula $\varphi(t_1, \ldots, t_\ell)$ is a \emph{template formula for $\Gamma$} if there is a function $\mathcal{M}$ such that (i) $\mathcal{M}$ is a surjective function from the set of interpretation functions to $\mathit{Pol}\left(\Gamma\right)$ and (ii) for any interpretation function $\mathfrak{I}$ and any request $(\domi{t}_1, \ldots, \domi{t}_\ell) \in \sortFont{T}^\mathfrak{S}_1 \times \ldots \times \sortFont{T}^\mathfrak{S}_\ell$, we have that $(\domi{t}_1, \ldots, \domi{t}_\ell) \in \varphi^\mathfrak{I}$ iff $\left(\mathcal{M}\left(\mathfrak{I}\right), \domi{t}_1, \ldots, \domi{t}_\ell\right) \in \mathit{assign}_{\Gamma}$.
\label{def:uni_template_formula}
\end{definition}

The mapping $\mathcal{M}$ provides a correspondence between interpretations and policies. $\mathcal{M}$ guarantees that each policy is represented by at least one interpretation. Therefore, we can search for an interpretation instead of a policy. For this reason, for the rest of the paper, we identify every formula in $\Lang$ with a \emph{policy language} and also refer to interpretation functions as \emph{policies}. 

%

\subsection{Formalizing the example}
\label{sub:example_rbac}

We now formally define the formula $\varphi^{\mathit{RBAC}}_N(\vari{u}, \vari{p}) \in \Lang$, introduced in Section~\ref{sub:motivating_example}, and show that it is a template formula for the language $\Gamma_N$ of all RBAC policies with at most $N$ roles.
Allowing a maximum number of roles is sufficient as one always can estimate a trivial bound on the maximum number of roles in an organization.


\paragraph{Template formula definition:} Consider a signature with a sort $\Roles$ denoting roles and with two (flexible) binary relation symbols $\mathit{UA}$ and $\mathit{PA}$ of types $\Users \times \Roles$ and $\Roles \times \Perms$, respectively. Define the formula
\begin{equation}
\varphi^{\mathit{RBAC}}_N(\vari{u}, \vari{p}) := \bigvee_{i \leq N} \left(\mathit{UA}(\vari{u},\ricons{r}_i) \land \mathit{PA}(\ricons{r}_i, \vari{p})\right).
\end{equation}
Here, $\ricons{r}_i$, for $1 \leq i \leq N$, is a rigid constant symbol of sort $\Roles$ (recall that we underline rigid symbols and denote constant symbols with \textsf{serif letters}). One could also use flexible constant symbols for roles, but, as we see later, the difficulty of implementing the policy miner increases with the number of flexible symbols.

\paragraph{Correctness proof:} We now define a mapping $\mathcal{M}$ that proves that $\varphi^{\mathit{RBAC}}_N$ is a template formula for $\Gamma_N$. For any interpretation function $\mathfrak{I}$, let $\mathcal{M}\left(\mathfrak{I}\right) = \left(U, \{\ricons{r}_1, \ldots, \ricons{r}_N\}, P, \mathit{UA}^{\mathfrak{I}}, \mathit{PA}^{\mathfrak{I}}\right).$
Observe that $\mathcal{M}\left(\mathfrak{I}\right)$ is an RBAC policy. Moreover, for $(\domi{u},\domi{p}) \in U \times P$, $(\domi{u},\domi{p}) \in \left(\varphi^{\mathit{RBAC}}_N\right)^\mathfrak{I}$ iff $(\domi{u},\domi{p}) \in \mathit{UA}^{\mathfrak{I}} \circ \mathit{PA}^{\mathfrak{I}}$ iff $(\pi_{\mathfrak{I}}, \domi{u}, \domi{p}) \in \mathit{assign}_{\mathit{RBAC}}$. It is also easy to prove that $\mathcal{M}$ is surjective on the set of all RBAC policies with at most $N$ roles. Hence, we can identify $\varphi^{\mathit{RBAC}}_N$ with the language of all RBAC policies with at most $N$ roles.\xqed{$\square$}

\begin{example}
To facilitate understanding $\mathcal{M}$, we show an RBAC policy $\pi$ and an interpretation function $\mathfrak{I}$ such that $\mathcal{M}\left(\mathfrak{I}\right) = \pi$.

Let $N = 2$ and assume that $U = \{\domi{Alice}, \domi{Bob}, \domi{Charlie}\}$ and that $P = \{\domi{c}, \domi{m}, \domi{d}\}$. The permissions in $P$ stand for ``create'', ``modify'', and ``delete''. Let $\domi{r}_1$ and $\domi{r}_2$ denote two roles. Consider the RBAC policy defined by Tables~\ref{tab:rbac_ua_1} and~\ref{tab:rbac_pa_1}.

\begin{table}[h!btp]%
\centering
\begin{minipage}{0.2\textwidth}%
\centering
\begin{tabular}{|c||c|c|}
\hline
& $\domi{r}_1$ & $\domi{r}_2$\\
\hline
\hline
$\domi{Alice}$ & $\times$ &\\
\hline
$\domi{Bob}$ & $\times$ &\\
\hline
$\domi{Charlie}$ & & $\times$ \\
\hline
\end{tabular}
\caption{User-assignment {relation}}
\label{tab:rbac_ua_1}
\end{minipage}%
\qquad
\begin{minipage}{0.2\textwidth}%
\centering
\begin{tabular}{|c||c|c|c|}
\hline
& $\domi{c}$ & $\domi{m}$ & $\domi{d}$\\
\hline
\hline
$\domi{r}_1$ & $\times$ & $\times$ &\\
\hline
$\domi{r}_2$ & & & $\times$ \\
\hline
\end{tabular}
\caption{Permission-assignment relation}
\label{tab:rbac_pa_1}
\end{minipage}%
%
\end{table}

We can define an interpretation function $\mathfrak{I}$ such that $\mathcal{M}\left(\mathfrak{I}\right)$ corresponds to the RBAC policy above. $\mathfrak{I}$ interprets the relation symbols $\mathit{UA}$ and $\mathit{PA}$ in the formula $\varphi^{\mathit{RBAC}}_N\left(\vari{u}, \vari{p}\right)$ as follows. For $\domi{u} \in U$ and $i \leq 2$, $\mathit{UA}^{\mathfrak{I}}\left(\domi{u},\ricons{r}_i\right)$ iff $(\domi{u},\domi{r}_i)$ is marked with an $\times$ in Table~\ref{tab:rbac_ua_1}. Similarly, for $\domi{p} \in P$ and $i \leq 2$, $\mathit{PA}^{\mathfrak{I}}\left(\ricons{r}_i, \domi{p}\right)$ iff $(\domi{r}_i, \domi{p})$ is marked with an $\times$ in Table~\ref{tab:rbac_pa_1}.
\label{ex:uni_rbac_template_fla}
\end{example}

%

\section{Probability distribution}
\label{sec:prob_dist}

\newcommand{\posprob}[2]{\mathbb{P}\left({#1} \mid #2\right)}

Let $\varphi \in \Lang$ be a policy language. We assume for the rest of the paper that $\varphi$ has two free variables $\vari{u}$ and $\vari{p}$ of sorts $\SortUsers$ and $\SortPermissions$, respectively. Our presentation extends in a straightforward way to more general cases.\looseness=-1

To design a policy miner using {\polymver} one must specify an \emph{objective function} $\lossfunc$. This is any function taking two inputs: a permission assignment $\Auth \subseteq U \times P$, which is a relation on $U$ and $P$ indicating what permissions each user has, and a policy $\mathfrak{I}$. An objective function outputs a value in $\mathbb{R}^+$ measuring how well $\varphi^{\mathfrak{I}}$ fits $\Auth$ and other policy requirements. The policy miner designer is in charge of specifying such a function. In Section~\ref{sec:examples}, we give other examples of objective functions.\looseness=-1

For illustration, consider the objective function

\begin{equation}
\lossfunc\left(\Auth, \mathfrak{I}; \varphi\right) = \sum_{(\domi{u},\domi{p}) \in U \times P} \left|\Auth(\domi{u},\domi{p}) - \varphi^{\mathfrak{I}}(\domi{u},\domi{p})\right|.
\label{eq:loss_func}
\end{equation}
Here, we identify the value 1 with the Boolean value \texttt{true} and the value 0 with the Boolean value \texttt{false}. Observe that $\lossfunc\left(\Auth, \mathfrak{I}; \varphi\right)$ is the size of the symmetric difference of the relations $\Auth$ and $\varphi^{\mathfrak{I}}$. Hence, lower values for $\lossfunc\left(\Auth, \mathfrak{I}; \varphi\right)$ are better.

The policy miners built with {\polymver} are \emph{probabilistic}. They receive as input a permission assignment $\Auth$ and compute a probability distribution over the set of all policies in a fixed policy language $\Gamma$. We use a Bayesian instead of a frequentist interpretation of probability. The probability of a policy $\mathfrak{I}$ does not measure how often $\mathfrak{I}$ is the outcome of an experiment, but rather how strong we believe $\mathfrak{I}$ to be the policy that decided the requests in $\Auth$.

We now define, given a permission assignment $\Auth$, a probability distribution $\posprob{\cdot}{\Auth}$ on policies. We first provide some intuition on $\posprob{\cdot}{\Auth}$'s definition and afterwards define it. For a permission assignment $\Auth$ and a policy $\mathfrak{I}$, we can see $\posprob{\cdot}{\Auth}$ as a quantity telling us how much we believe $\mathfrak{I}$ to be the organization's policy, given that $\Auth$ is the organization's permission assignment.\looseness=-1

Policy miners receive as input a permission assignment $\Auth$ and then search for a policy $\mathfrak{I}^*$ that maximizes $\posprob{\cdot}{\Auth}$. Here, $\posprob{\cdot}{\Auth}$ is defined as the ``most general'' distribution that fulfills the following requirement: \emph{for any policy $\mathfrak{I}$, the lower $\lossfunc\left(\Auth, \mathfrak{I}; \varphi\right)$ is, the more likely $\mathfrak{I}$ is.} Following the principle of maximum entropy~\cite{kesavan1990maximum}, the most general distribution that achieves this is
\begin{equation}
\posprob{\mathfrak{I}}{\Auth} = \frac{\exp\left(- \beta L(\Auth, \mathfrak{I}; \varphi)\right)}{\sum_{\mathfrak{I}'} \exp\left(- \beta L(\Auth, \mathfrak{I}'; \varphi)\right)},
\label{eq:prob_dist}
\end{equation}
where $\mathfrak{I}'$ ranges over all policies. Recall that we consider only finite structures. Hence, all our carrier sets are finite, so there are only finitely many policies.

The value $\beta > 0$ is a parameter that the policy miner varies during the search for the most likely policy. The search uses deterministic annealing, an optimization procedure inspired by simulated annealing~\cite{kirkpatrick1983optimization, rose1992vector, rose1998deterministic}. In our case, it initially sets $\beta$ to a very low value, so that all policies are almost equally likely. Then it gradually increases $\beta$ while, at the same time, searching for the most likely policy. As $\beta$ increases, those policies that minimize $L(\Auth, \cdot\,; \varphi)$ become more likely. In this way, deterministic can escape from low-quality local maxima of $\posprob{\cdot}{\Auth}$. When $\beta \to \infty$, only those policies that minimize $L(\Auth, \cdot\,; \varphi)$ have a positive probability and the search converges to a local maximum of $\posprob{\cdot}{\Auth}$.\looseness=-1
%

We now define the probability distribution given in Equation~\ref{eq:prob_dist}.

\begin{definition}
For a formula $\varphi \in \Lang$, we define the probability space $\mathfrak{P}_\varphi = \left(\Omega, 2^{\Omega}, \posprob{\cdot}{\Auth}\right)$ as follows. 
\begin{itemize}
\item $\Omega$ is the set of all interpretation functions (i.e., policies).
\item $2^\Omega$ is the set of all subsets of $\Omega$. Since all carrier sets of all sorts are finite (Definition~\ref{def:bkg_sigma_structure}), $\Omega$ and $2^\Omega$ are finite.

\item For $\mathfrak{I} \in \Omega$,
\vspace{-5pt}
\begin{equation}
\posprob{\mathfrak{I}}{\Auth} = \frac{\exp\left(-\beta L(\Auth, \mathfrak{I}; \varphi)\right)}{\sum_{\mathfrak{I}'} \exp\left(-\beta L(\Auth, \mathfrak{I}'; \varphi)\right)}.
\end{equation}
Finally, for $O \in 2^\Omega$, let $\posprob{O}{\Auth} = \sum_{\mathfrak{I} \in O}\posprob{\mathfrak{I}}{\Auth}$.
\end{itemize}
\label{def:uni_gibbs_prob}
\end{definition}

The following theorem proves that $\posprob{\cdot}{\Auth}$ is the ``most general'' distribution that fulfills the requirement mentioned above. More precisely, $\posprob{\cdot}{\Auth}$ is the maximum-entropy probability distribution 
where the probability of a policy $\mathfrak{I}$ increases whenever $\lossfunc(\Auth, \mathfrak{I}; \varphi)$ decreases~\cite{jaynes1957information,tikochinsky1984alternative}.\looseness=-1

\begin{theorem}
$\posprob{\cdot}{\Auth}$ is the distribution $P$ on policies that maximizes $P$'s entropy and is subject to the following constraints.
\begin{itemize}
\item $\sum_{\mathfrak{I}} P\left(\mathfrak{I}\right)\lossfunc(\Auth, \mathfrak{I}; \varphi) \leq \ell$, for some fixed bound $\ell$.
\item If $\beta > 0$, then $P\left(\mathfrak{I}\right) > P\left(\mathfrak{I}'\right)$, for any two policies $\mathfrak{I}$ and $\mathfrak{I}'$ with $\lossfunc(\Auth, \mathfrak{I}; \varphi) < \lossfunc(\Auth, \mathfrak{I}'; \varphi)$.
\end{itemize}
\label{thm:uni_max_entrop_dist}
\end{theorem}

\begin{proof} 
It suffices to drop the second constraint and use Lagrange multipliers to verify that $\posprob{\mathfrak{I}}{\Auth}$ is the optimal distribution. Observe that $\posprob{\mathfrak{I}}{\Auth}$ satisfies the second constraint.
\end{proof}

\newcommand{\factorpdf}[2]{q_{#1} \left(#2\right)}
\newcommand{\bestfpdfpram}[1]{\widehat{q}_{#1}}
\newcommand{\meanfield}[3]{\mathbb{E}_{\substack{#2}}{\left[#3\right]}}
\newcommand{\mnfld}[2]{E_{#1, #2}}
\newcommand{\atomics}[1]{\mathcal{A}\left(#1\right)}
\newcommand{\atomfacts}[2]{\mathfrak{F}\left(#1\right)}



\begin{example}
We illustrate the probability distribution defined above for the language of all RBAC policies with at most $N$ roles, defined in Section~\ref{sub:example_rbac}. For simplicity, we fix $N = 2$ in this example. Assume that $U = \{\domi{Alice}, \domi{Bob}, \domi{Charlie}\}$ and that $P = \{\domi{c}, \domi{m}, \domi{d}\}$, as defined in Example~\ref{ex:uni_rbac_template_fla}. Assume given a permission assignment $\mathit{Auth}$ and two policies $\mathfrak{I}_1$ and $\mathfrak{I}_2$ as shown in Tables~\ref{tab:auth_1}--\ref{tab:phi_i_2}. 

Recall that $\left(\varphi^{\mathit{RBAC}}_N\right)^{\mathfrak{I}_1}$ and $\left(\varphi^{\mathit{RBAC}}_N\right)^{\mathfrak{I}_2}$ are the permission assignments induced by $\mathfrak{I}_1$ and $\mathfrak{I}_2$, respectively. Observe that $\left(\varphi^{\mathit{RBAC}}_N\right)^{\mathfrak{I}_1}$ and $\Auth$ differ by one entry, whereas $\left(\varphi^{\mathit{RBAC}}_N\right)^{\mathfrak{I}_2}$ and $\Auth$ differ by two. Hence, 
%
$\lossfunc\left(\Auth, \mathfrak{I}_1; \varphi^{\mathit{RBAC}}_N\right) = 1 < 2 = \lossfunc\left(\Auth, \mathfrak{I}_2; \varphi^{\mathit{RBAC}}_N\right).$
%
As a result, for any $\beta > 0$, we get that
%
$\posprob{\mathfrak{I}_1}{\Auth} = \frac{\exp\left(-\beta\right)}{Z} > \frac{\exp\left(-2\beta\right)}{Z} = \posprob{\mathfrak{I}_2}{\Auth},$
%
where $Z = \sum_{\mathfrak{I}'} \exp\left(-\beta L(\Auth, \mathfrak{I}'; \varphi)\right)$.
\label{ex:uni_rbac_fla}
\end{example}

\newcommand{\colwidth}{0.2\textwidth}
\newcommand{\readaction}{$\domi{c}$}
\newcommand{\writeaction}{$\domi{m}$}
\newcommand{\createaction}{$\domi{d}$}

\begin{table}[h!btp]%
\centering
\begin{minipage}{\colwidth}%
\centering
\begin{tabular}{|c||c|c|c|}
\hline
& \readaction & \writeaction & \createaction\\
\hline
\hline
$\domi{Alice}$ & $\times$ & $\times$ &\\
\hline
$\domi{Bob}$ & $\times$ & $\times$ &\\
\hline
$\domi{Charlie}$ & $\times$ & & $\times$ \\
\hline
\end{tabular}
\caption{$\Auth$}
\label{tab:auth_1}
\end{minipage}%
\end{table}
\begin{table}[h!btp]%
\centering
\begin{minipage}{\colwidth}%
\centering
\begin{tabular}{|c||c|c|}
\hline
& $\ricons{r}_1^{\mathfrak{I}_1}$ & $\ricons{r}_2^{\mathfrak{I}_1}$\\
\hline
\hline
$\domi{Alice}$ & $\times$ &\\
\hline
$\domi{Bob}$ & $\times$ &\\
\hline
$\domi{Charlie}$ & & $\times$ \\
\hline
\end{tabular}
\caption{$\mathit{UA}^{\mathfrak{I}_1}$}
\label{tab:ua_i_1}
\end{minipage}%
\begin{minipage}{\colwidth}%
\centering
\begin{tabular}{|c||c|c|}
\hline
& $\ricons{r}_1^{\mathfrak{I}_2}$ & $\ricons{r}_2^{\mathfrak{I}_2}$\\
\hline
\hline
$\domi{Alice}$ & $\times$ &\\
\hline
$\domi{Bob}$ & $\times$ &\\
\hline
$\domi{Charlie}$ & & $\times$ \\
\hline
\end{tabular}
\caption{$\mathit{UA}^{\mathfrak{I}_2}$}
\label{tab:ua_i_2}
\end{minipage}%
%
\end{table}
\begin{table}[h!btp]%
\centering
\begin{minipage}{\colwidth}%
\centering
\begin{tabular}{|c||c|c|c|}
\hline
& \readaction & \writeaction & \createaction\\
\hline
\hline
$\ricons{r}_1^{\mathfrak{I}_1}$ & $\times$ & $\times$ &\\
\hline
$\ricons{r}_2^{\mathfrak{I}_1}$ & & & $\times$ \\
\hline
\end{tabular}
\caption{$\mathit{PA}^{\mathfrak{I}_1}$}
\label{tab:pa_i_1}
\end{minipage}%
\begin{minipage}{\colwidth}%
\centering
\begin{tabular}{|c||c|c|c|}
\hline
& \readaction & \writeaction & \createaction\\
\hline
\hline
$\ricons{r}_1^{\mathfrak{I}_2}$ & & $\times$ &\\
\hline
$\ricons{r}_2^{\mathfrak{I}_2}$ & $\times$ & & $\times$ \\
\hline
\end{tabular}
\caption{$\mathit{PA}^{\mathfrak{I}_2}$}
\label{tab:pa_i_2}
\end{minipage}
%
\end{table}
\begin{table}[h!btp]
\centering
\begin{minipage}{\colwidth}%
\centering
\begin{tabular}{|c||c|c|c|}
\hline
& \readaction & \writeaction & \createaction\\
\hline
\hline
$\domi{Alice}$ & $\times$ & $\times$ &\\
\hline
$\domi{Bob}$ & $\times$ & $\times$ &\\
\hline
$\domi{Charlie}$ & & & $\times$\\
\hline
\end{tabular}
\caption{$\left(\varphi^{\mathit{RBAC}}_N\right)^{\mathfrak{I}_1}$}
\label{tab:phi_i_1}
\end{minipage}%
\qquad
\begin{minipage}{\colwidth}%
\centering
\begin{tabular}{|c||c|c|c|}
\hline
& \readaction & \writeaction & \createaction\\
\hline
\hline
$\domi{Alice}$ & & $\times$ &\\
\hline
$\domi{Bob}$ & & $\times$ &\\
\hline
$\domi{Charlie}$ & $\times$ & & $\times$\\
\hline
\end{tabular}
\caption{$\left(\varphi^{\mathit{RBAC}}_N\right)^{\mathfrak{I}_2}$}
\label{tab:phi_i_2}
\end{minipage}
\end{table}


\section{Applying mean-field approximation}
\label{sec:build_var_infer}

The policy miner that is built with {\polymver} is an algorithm that receives as input a permission assignment $\Auth$ and computes a policy $\mathfrak{I}$ that approximately maximizes $\posprob{\cdot}{\Auth}$, while letting $\beta \to \infty$. Since computing $\posprob{\cdot}{\Auth}$ is intractable, we use \emph{mean-field approximation}~\cite{bishop2006pattern}, a technique that defines an iterative procedure to approximate $\posprob{\cdot}{\Auth}$ with a distribution $q\left(\cdot\right)$. It turns out that computing and maximizing $q(\cdot)$ is much easier than computing and maximizing $\posprob{\cdot}{\Auth}$. The policy miner is then an algorithm implementing the computation of $q$ and its maximization.

We next introduce some random variables that help to measure the probability that a policy authorizes a particular request $(\domi{u},\domi{p}) \in U \times P$ (Section~\ref{sub:aux_defs}). Afterwards, we present the approximating distribution $q$ (Section~\ref{sub:approx_conditional}).

\subsection{Random variables}
\label{sub:aux_defs}

\newcommand{\subflarv}[4]{\lceil #2^{\mathfrak{X}\langle #1\rangle}\left(#3, #4\right)\rceil}
\newcommand{\termrv}[2]{\lceil #2^{\mathfrak{X}\langle #1\rangle}\rceil}
\newcommand{\evtermrv}[4]{\lceil #2^{\mathfrak{X}\langle #1\rangle}\left(#3, #4\right)\rceil}
\newcommand{\rangerandfact}[1]{\mathit{Range}\left(#1\right)}

Recall that the sample space $\Omega$ of the distribution $\posprob{\cdot}{\Auth}$ from Definition~\ref{def:uni_gibbs_prob} is the set of all policies $\mathfrak{I}$. Let $\mathfrak{X}$ be a random variable mapping $\mathfrak{I} \in \Omega$ to $\mathfrak{I}$. Although $\mathfrak{X}$'s definition is trivial, it will help us to understand other random variables that we introduce later. We can understand $\mathfrak{X}$ as an ``unknown policy'' and, for a policy $\mathfrak{I}$, the probability statement $\posprob{\mathfrak{X} = \mathfrak{I}}{\Auth}$ measures how much we believe that $\mathfrak{X}$ is actually $\mathfrak{I}$, given that the organization's permission assignment is $\Auth$. By definition, $\posprob{\mathfrak{X} = \mathfrak{I}}{\Auth} = \posprob{\mathfrak{I}}{\Auth}$.

%
%

\begin{definition}
Let $\varphi \in \Lang$ and let $W$ be a flexible relation symbol occurring in $\varphi$ of type $\sortFont{S}_1 \times \ldots \times \sortFont{S}_k$ and let $f$ be a flexible function symbol occurring in $\varphi$ of type $\sortFont{S}_1 \times \ldots \times \sortFont{S}_{k} \to \sortFont{S}$. Let $(\domi{a}_1, \ldots, \domi{a}_k) \in \sortFont{S}^{\mathfrak{S}}_1 \times \ldots \times \sortFont{S}^{\mathfrak{S}}_k$. Recall that $\mathfrak{S}$ maps sorts to carrier sets. We define the random variable $W^{\mathfrak{X}}\left(\domi{a}_1, \ldots, \domi{a}_{k}\right) : \Omega \to \{0, 1\}$ that maps $(\Auth, \mathfrak{I}) \in \Omega$ to $W^{\mathfrak{I}}\left(\domi{a}_1, \ldots, \domi{a}_{k}\right) \in \{0, 1\}$. Similarly, we define the random variable $f^{\mathfrak{X}}\left(\domi{a}_1, \ldots, \domi{a}_{k}\right) : \Omega \to \sortFont{S}^{\mathfrak{S}}$ that maps $(\Auth, \mathfrak{I}) \in \Omega$ to $f^{\mathfrak{I}}\left(\domi{a}_1, \ldots, \domi{a}_{k}\right) \in \sortFont{S}^{\mathfrak{S}}$. We call these random variables \emph{random facts of $\varphi$}.
\label{def:random_fact}
\end{definition}

\begin{example}
Let us examine some random facts of the formula $\varphi^{\mathit{RBAC}}_N$ from Example~\ref{ex:uni_rbac_fla}. One such random fact is $\mathit{UA}^{\mathfrak{X}}\left(\domi{Alice}, \ricons{r}_1\right)$, which can take the values 0 and 1, so $\mathit{UA}^{\mathfrak{X}}\left(\domi{Alice}, \ricons{r}_1\right)$ is a Bernoulli random variable whose probability distribution is defined by
\begin{align}
\nonumber&\posprob{\mathit{UA}^{\mathfrak{X}}\left(\domi{Alice}, \ricons{r}_1\right) = 1}{\Auth} =\\
&\qquad \posprob{\left\{ \mathfrak{I} \in \Omega \mid \mathit{UA}^{\mathfrak{I}}\left(\domi{Alice}, \ricons{r}_1\right) = 1\right\}}{\Auth}.
\end{align}
More generally, the set of random facts for $\varphi^{\mathit{RBAC}}_N$ is

\begin{align}
\nonumber
\atomfacts{\varphi^{\mathit{RBAC}}_N}{\mathfrak{I}} = &\left\{\mathit{UA}^{\mathfrak{X}}\left(\domi{u}, \ricons{r}_i\right) \mid \domi{u} \in U, i \leq N\right\} \cup \\
&\left\{\mathit{PA}^{\mathfrak{X}}\left(\ricons{r}_i, \domi{p}\right) \mid \domi{p} \in P, i \leq N\right\}.
\end{align}

If we set $N = 2$ and replace each random fact with a Boolean value, as indicated by Tables~\ref{tab:ua_i_1} and~\ref{tab:ua_i_2}, then we get an RBAC policy.\looseness=-1

Just like a statement of the form $\posprob{\mathfrak{X} = \mathfrak{I}}{\Auth}$ quantifies how much we believe that $\mathfrak{X} = \mathfrak{I}$ for a given $\Auth$, a statement of the form $\posprob{\mathit{UA}^{\mathfrak{X}}\left(\domi{Alice}, \ricons{r}_1\right) = 1}{\Auth}$ quantifies how much we believe that role $\ricons{r}_1$ is assigned to $\domi{Alice}$ for a given $\Auth$.
\end{example}

%
%

\begin{observation}
Since we assume carrier sets to be finite, a random fact always has a discrete distribution. In particular, random facts built from flexible relation symbols have Bernoulli distributions as they can only take Boolean values.\hfill$\square$
\label{obs:random_fact_disc_distr}
\end{observation}

We usually denote random facts with Fraktur letters $\mathfrak{f}, \mathfrak{g}, \ldots$
For a random fact $\mathfrak{f}$ of the form $W^{\mathfrak{X}}\left(\domi{a}_1, \ldots, \domi{a}_k\right)$, we denote by $\mathfrak{f}^{\mathfrak{I}}$ the Boolean value $W^{\mathfrak{I}}\left(\domi{a}_1, \ldots, \domi{a}_k\right)$. Similarly, when $\mathfrak{f}$ is of the form $f^{\mathfrak{X}}\left(\domi{a}_1, \ldots, \domi{a}_k\right)$, we denote by $\mathfrak{f}^{\mathfrak{I}}$ the value $f^{\mathfrak{I}}\left(\domi{a}_1, \ldots, \domi{a}_k\right)$. Finally, we denote $\mathfrak{f}$'s range with $\rangerandfact{\mathfrak{f}}$.\looseness=-1

For a policy language $\varphi \in \Lang$, we denote by $\atomfacts{\varphi}{\mathfrak{I}}$ the set of all random facts of $\varphi$. Recall that we assume all our carrier sets to be finite, so $\atomfacts{\varphi}{\mathfrak{I}}$ is finite.

Observe that, for any formula $\varphi \in \Lang$, replacing each random fact $\mathfrak{f}$ in $\atomfacts{\varphi}{\mathfrak{I}}$ with a value in $\rangerandfact{\mathfrak{f}}$ yields a policy. Hence, a policy miner, instead of searching for a policy $\mathfrak{I}$, it just searches for adequate values for all random facts in $\atomfacts{\varphi}{}$. We formalize this in Lemma~\ref{lem:cond_prob_is_vec_prob}, whose proof is in Appendix~\ref{sec:proof_cond_prob_is_vec_prob}.

\begin{lemma} For a policy language $\varphi \in \Lang$,
\begin{equation}
\posprob{\mathfrak{I}}{\Auth} = \posprob{\left(\mathfrak{f}^{\mathfrak{X}}\right)_{\mathfrak{f} \in \atomfacts{\varphi}{}} = \left(\mathfrak{f}^{\mathfrak{I}}\right)_{\mathfrak{f} \in \atomfacts{\varphi}{}}}{\Auth}.
\end{equation}
\label{lem:cond_prob_is_vec_prob}
\end{lemma}

We denote by $h\left(\cdot\right)$ the function $\posprob{\left(\mathfrak{f}^{\mathfrak{X}}\right)_{\mathfrak{f} \in \atomfacts{\varphi}{}} = \cdot}{\Auth}$. To avoid cluttered notation, we write $h\left(\mathfrak{I}\right)$ instead of $h\left(\left(\mathfrak{f}^{\mathfrak{I}}\right)_{\mathfrak{f} \in \atomfacts{\varphi}{}}\right)$.

We conclude this section by defining some other useful random variables. Recall that $\mathfrak{X}$ is the random variable that maps $\mathfrak{I} \in \Omega$ to $\mathfrak{I}$.\looseness=-1

\begin{definition}
For $(\domi{u},\domi{p}) \in U \times P$, $\varphi \in \Lang$, we define the random variable $\varphi^{\mathfrak{X}}\left(\domi{u},\domi{p}\right) : \Omega \to \{0, 1\}$ as the function mapping $(\Auth, \mathfrak{I})$ to $\varphi^{\mathfrak{I}}\left(\domi{u},\domi{p}\right)$.
\label{def:random_formula}
\end{definition}

\begin{definition}
For $\varphi \in \Lang$, $\Auth \subseteq U \times P$, we define the following random variable:
\begin{equation}
\lossfunc\left(\Auth, \mathfrak{X}; \varphi\right) := \sum_{(\domi{u},\domi{p}) \in U \times P} \left|\Auth(\domi{u},\domi{p}) - \varphi^{\mathfrak{X}}(\domi{u},\domi{p})\right|.
\label{eq:rv_loss_func}
\end{equation}
\end{definition}

\subsection{Approximating the distribution}
\label{sub:approx_conditional}

A mean-field approximation of the probability distribution $h$ is a distribution $q$ defined by
\begin{equation}
q\left(\mathfrak{I}\right) := \prod_{\mathfrak{f} \in \atomfacts{\varphi}{\mathfrak{I}}} \factorpdf{\mathfrak{f}}{\mathfrak{f}^\mathfrak{I}},
\label{eq:post_approx_dist}
\end{equation}
%
where $q_{\mathfrak{f}}: \rangerandfact{\mathfrak{f}} \to [0, 1]$ is a probability mass function for $\mathfrak{f}$. Hence, $\sum_{\domi{b} \in \rangerandfact{\mathfrak{f}}} \factorpdf{\mathfrak{f}}{\domi{b}} = 1$. For $\domi{b} \in \rangerandfact{\mathfrak{f}}$, the value $\factorpdf{\mathfrak{f}}{\domi{b}}$ denotes the probability, according to $q_{\mathfrak{f}}$, that $\mathfrak{f} = \domi{b}$.\looseness=-1


Observe that $q\left(\mathfrak{I}\right)$'s factorization implies that the set of random facts is mutually independent. This is not true in general, as $h$ may not be necessarily factorized like $q$. This independence assumption is imposed by mean-field theory to facilitate computations. Our experimental results in Section~\ref{sec:experiments} show that, despite this approximation, we still mine high quality policies.\looseness=-1

According to mean-field theory, the distributions $\left\{\bestfpdfpram{\mathfrak{f}} \mid \mathfrak{f} \in \atomfacts{\varphi}{\mathfrak{I}}\right\}$ that make $q$ best approximate $h$ are given by
\begin{equation}
	\bestfpdfpram{\mathfrak{f}}\left(\domi{b}\right) 	
	= \frac{\exp\left(
		-\beta \meanfield{V \neq \mathfrak{f}}{\mathfrak{f} \mapsto \domi{b}}
			{\lossfunc\left(\Auth, \mathfrak{X}; \varphi\right)}
		\right)}
		{\sum_{\domi{b}' \in \rangerandfact{\mathfrak{f}}}\exp\left(
		-\beta \meanfield{V \neq \mathfrak{f}}{\mathfrak{f} \mapsto \domi{b}'} 
			{\lossfunc\left(\Auth, \mathfrak{X}; \varphi\right)}
		\right)},
\label{eq:updatebestfpdf}
\end{equation}
where $\domi{b} \in \rangerandfact{\mathfrak{f}}$ and $\meanfield{V \neq \mathfrak{f}}{\mathfrak{f} \mapsto \domi{b}} {\lossfunc\left(\Auth, \mathfrak{X}; \varphi\right)}$ is the expectation of $\lossfunc\left(\Auth, \mathfrak{X}; \varphi\right)$ after replacing every occurrence of the random fact $\mathfrak{f}$ with $\domi{b}$~\cite{bishop2006pattern}. This expectation is computed using the distribution $q$. Therefore,
\begin{equation}
\begin{array}{l}
\meanfield{V \neq \mathfrak{f}}{\mathfrak{f} \mapsto \domi{b}} {\lossfunc\left(\Auth, \mathfrak{X}; \varphi\right)}\\[5pt] 
\qquad  = \displaystyle\sum_{\mathfrak{I}} \displaystyle\prod_{\substack{\mathfrak{g} \in \atomfacts{\varphi}{\mathfrak{I}}\\\mathfrak{g} \neq \mathfrak{f}}} \bestfpdfpram{\mathfrak{g}}\left(\mathfrak{g}^\mathfrak{I}\right) \left(\lossfunc\left(\Auth, \mathfrak{I}; \varphi\right)\left\{\mathfrak{f} \mapsto \domi{b}\right\}\right).
\end{array}
\label{eq:mf_approx}
\end{equation}
%
%
%
Here, $\lossfunc\left(\Auth, \mathfrak{I}; \varphi\right)\left\{\mathfrak{f} \mapsto \domi{b}\right\}$ is obtained from $\lossfunc\left(\Auth, \mathfrak{I}; \varphi\right)$ by replacing $\mathfrak{f}$ with $\domi{b}$. 

Using Lemma~\ref{lem:cond_prob_is_vec_prob} and the distribution $q$, we can approximate $\arg\max_{\mathfrak{I}} \prob\left(\mathfrak{I} \mid \Auth\right)$ by maximizing $q$.

\begin{observation}
%
$\max_{\mathfrak{I}}\posprob{\mathfrak{I}}{\Auth} = \max_{\mathfrak{I}}h\left(\mathfrak{I}\right) \approx \max_{\mathfrak{I}}q\left(\mathfrak{I}\right).$
\end{observation}

The desired miner is then an algorithm that computes $q$, while letting $\beta \to \infty$, and then computes the policy $\mathfrak{I}^*$ that maximizes $q$.

\section{Building the policy miner}
\label{sec:polymer}

To compute $q$, as given by Equation~\ref{eq:post_approx_dist}, the desired policy miner could use Equation~\ref{eq:updatebestfpdf} to compute $\bestfpdfpram{\mathfrak{f}}$, for each $\mathfrak{f} \in \atomfacts{\varphi}{\mathfrak{I}}$. Observe, however, that Equation~\ref{eq:updatebestfpdf} is recursive, since the computation of the expectations on the right hand side requires $\left\{\bestfpdfpram{\mathfrak{f}} \mid \mathfrak{f} \in \atomfacts{\varphi}{\mathfrak{I}}\right\}$, as indicated by Equation~\ref{eq:mf_approx}. This recursive dependency is handled by iteratively computing, for each $\mathfrak{f} \in \atomfacts{\varphi}{\mathfrak{I}}$, a function $\tilde{q_\mathfrak{f}}$ that approximates $\bestfpdfpram{\mathfrak{f}}$~\cite{bishop2006pattern}. We illustrate this in the step~\ref{step:pram_update} below.

Algorithm~\ref{algo:varinfer} gives the pseudocode for computing and maximizing $q$, which is the essence of the desired policy miner. We give next an overview.

\begin{enumerate}
\item \textbf{Initialization (lines~\ref{state:init}--\ref{state:init_setup}). }Each distribution $\tilde{q_{\mathfrak{f}}}$ is randomly set to an arbitrary function such that $\sum_{\domi{b}} \tilde{q_{\mathfrak{f}}}(\domi{b}) = 1$.
\item \textbf{Update loop (lines~\ref{state:for_start}--\ref{state:for_end}). }We perform a sequence of iterations that update $\left\{\tilde{q_\mathfrak{f}} \mid \mathfrak{f} \in \atomfacts{\varphi}{\mathfrak{I}}\right\}$ and $\beta$. The number $T$ of iterations is fixed before execution.
\begin{enumerate}
\item\label{step:pram_update} \textbf{Parameter update (line~\ref{state:for_fact}--\ref{state:end_for_fact}). }At each iteration, we compute a random ordering $\mathit{RS}(\atomfacts{\varphi}{\mathfrak{I}})$ of all the random facts. Then, for each $\mathfrak{f}$ in that order, $\tilde{q_{\mathfrak{f}}}$ is updated to the right-hand side of Equation~\ref{eq:updatebestfpdf} (lines~\ref{state:second_for_b}--\ref{state:update_pram}), 
but instead of using $\left\{\bestfpdfpram{f} \mid \mathfrak{f} \in \atomfacts{\varphi}{\mathfrak{I}}\right\}$, we use $\left\{\tilde{q_{\mathfrak{f}}} \mid \mathfrak{f} \in \atomfacts{\varphi}{\mathfrak{I}}\right\}$ to compute the expectations.
\item \textbf{Hyper-parameter update (line~\ref{state:update_beta}). }After each iteration, we increase $\beta$ by a factor of $\alpha$, defined before execution. This approach, originally defined for deterministic annealing, avoids that the algorithm is trapped in a bad local maximum in the early iterations~\cite{rose1992vector,rose1998deterministic}. 
\end{enumerate}
\item \textbf{Policy computation (line~\ref{state:compute_pol}). }Finally, we compute the policy $\mathfrak{I}^* = \arg\max_{\mathfrak{I}}q\left(\mathfrak{I}\right)$. By looking at Equation~\ref{eq:post_approx_dist}, we see that to maximize $q$, it suffices to maximize $q_\mathfrak{f}$, for every $\mathfrak{f} \in \atomfacts{\varphi}{\mathfrak{I}}$. Hence, we let $\mathfrak{I}^*$ be the policy that satisfies $\mathfrak{f}^{\mathfrak{I}^*} = \arg\max_{\domi{b} \in \rangerandfact{\mathfrak{f}}}\tilde{q}_{\mathfrak{f}}(\domi{b})$.
\end{enumerate}


%
\begin{algorithm}[h!]
\SetStartEndCondition{ }{}{}%
\SetKwProg{Fn}{}{\string:}{}
\SetKwFor{For}{for}{\string:}{}%
\SetKwIF{If}{ElseIf}{Else}{if}{:}{elif}{else:}{}%
\DontPrintSemicolon
\LinesNumbered
\Fn(){\textsc{PolicyMiner}$\left(L, \Auth, \varphi, \alpha, \beta, T\right)$}{
\For{$\mathfrak{f} \in \atomfacts{\varphi}{\mathfrak{I}}$\label{state:init}}
{\text{Randomly initialize $\tilde{q}_{\mathfrak{f}}$.}}
\label{state:init_setup}
\For{$i = 1 \ldots T$\label{state:for_start}}
{\For{$\mathfrak{f} \in \mathit{RS}\left(\atomfacts{\varphi}{\mathfrak{I}}\right)$\label{state:for_fact}}
{
\For{$\domi{b} \in \rangerandfact{\mathfrak{f}}$\label{state:second_for_b}}
{$\tilde{q}_{\mathfrak{f}}\left(\domi{b}\right) \gets \displaystyle\frac{\exp\left(-\beta\meanfield{V \neq W}{\mathfrak{f} \mapsto \domi{b}}{\lossfunc\left(\Auth, \mathfrak{X}; \varphi\right)}\right)}{\sum_{\domi{b}'}\exp\left(-\beta\meanfield{V \neq W}{\mathfrak{f} \mapsto \domi{b}'}{\lossfunc\left(\Auth, \mathfrak{X}; \varphi\right)}\right)}.$}\label{state:update_pram}}\label{state:end_for_fact}
$\beta \gets \alpha \times \beta.$\label{state:update_beta}}\label{state:for_end}
Define $\mathfrak{I}^*$ by letting $\mathfrak{f}^{\mathfrak{I}^*} = \arg\max_{\domi{b}}\tilde{q}_{\mathfrak{f}}(\domi{b})$, for $\mathfrak{f} \in \atomfacts{\varphi}{}$.\label{state:compute_pol}\\
\textbf{return} $\mathfrak{I}^*$.}
\caption{The policy miner.}
\label{algo:varinfer}
\end{algorithm}

Observe that the policy miner requires values for the hyper-parameters $\alpha$, $\beta$, and $T$ as input. Adequate values can be computed using machine-learning methods like grid search~\cite{hyperparamsearch}, which we briefly recall in Appendix~\ref{sub:grid_search}.

\subsection{Simplifying the computation of expectations}
\label{sub:simplifying}

One need not be knowledgeable about deterministic annealing or mean-field approximations to implement Algorithm~\ref{algo:varinfer} in a standard programming language. The only part requiring knowledge in probability theory is the computation of the expectations in line~\ref{state:update_pram}. We now define the notion of \emph{diverse} random variables and show that expectations of some diverse random variables can easily be computed recursively using some basic equalities.

\begin{definition}
A random variable $X$ is \emph{diverse} if (i) it can be constructed from constant values and random facts using only arithmetic and Boolean operations and (ii) any random fact is used in the construction at most once.
\label{def:diverse_rv}
\end{definition}

\begin{example}
Let $(\domi{u}, \domi{p}) \in U \times P$ and let $V$, $W$, and $Y$ be flexible relation symbols. Then $V^\mathfrak{X}\left(\domi{u}, \domi{p}\right) + W^\mathfrak{X}\left(\domi{u}, \domi{p}\right)$ is diverse, but $V^\mathfrak{X}\left(\domi{u}, \domi{p}\right)W^\mathfrak{X}\left(\domi{u}, \domi{p}\right) + W^\mathfrak{X}\left(\domi{u}, \domi{p}\right)Y^\mathfrak{X}\left(\domi{u}, \domi{p}\right) + V^\mathfrak{X}\left(\domi{u}, \domi{p}\right)Y^\mathfrak{X}\left(\domi{u}, \domi{p}\right)$ is not, since each random fact there occurs more than once.
\end{example}

\begin{corollary}
Let $\varphi \in \Lang$ and $(\domi{u}, \domi{p}) \in U \times P$, then $\varphi^{\mathfrak{X}}\left(\domi{u}, \domi{p}\right)$ is diverse iff every atomic formula that occurs in $\varphi$ occurs exactly once.
\label{cor:occurs_once}
\end{corollary}

This corollary is a direct consequence of Definition~\ref{def:diverse_rv}. Observe that, for $\varphi \in \Lang$, one can check in time linear in $\varphi$'s length that every atomic formula occurring in $\varphi$ occurs exactly once.

\begin{example}
Recall the formula $\varphi^{\mathit{RBAC}}_{N}$ defined in Section~\ref{sub:example_rbac}. Observe that each atomic formula occurs exactly once. Hence, for $(\domi{u}, \domi{p}) \in U \times P$, the random variable $\left(\varphi^{\mathit{RBAC}}_{N}\right)^{\mathfrak{X}}(\domi{u}, \domi{p})$ is diverse. 
\label{ex:rbac_abac_diverse}
\end{example}

The following lemma, proved in Appendix~\ref{sec:app_uni_expectation_computation}, shows how to recursively compute $\meanfield{V \neq \mathfrak{f}}{\mathfrak{f} \mapsto \domi{b}}{\lossfunc\left(\Auth, \mathfrak{X}; \varphi\right)}$ when $\varphi^\mathfrak{X}(\domi{u},\domi{p})$ is diverse.

\begin{lemma} Let $\mathfrak{f}$ and $\mathfrak{g}$ be facts, $\varphi$ be a formula in $\Lang$, $(\domi{u},\domi{p}) \in U \times P$, and 
 $\{\psi_i\}_i \subseteq \Lang$. Assume that $\varphi^{\mathfrak{X}}(\domi{u},\domi{p})$ and $\left(\bigwedge_i \psi_i\right)^{\mathfrak{X}}(\domi{u},\domi{p})$ are diverse. Then the following equalities hold.
%
\begin{align*}
\meanfield{V \neq \mathfrak{f}}{\mathfrak{f} \mapsto \domi{b}} {\mathfrak{g}} &= 
\begin{cases}
\domi{b} & \text{ if $\mathfrak{f} = \mathfrak{g}$ and}\\
\sum_{\domi{b} \in \rangerandfact{\mathfrak{g}}}\tilde{q}_{\mathfrak{g}}\left(\domi{b}\right) \domi{b} & \text{ otherwise. }
\end{cases}\\
\meanfield{V \neq \mathfrak{f}}{\mathfrak{f} \mapsto \domi{b}} {\left(\neg \varphi\right)^\mathfrak{X}(\domi{u},\domi{p})} &= 1 - \meanfield{V \neq \mathfrak{f}}{\mathfrak{f} \mapsto \domi{b}} {\varphi^\mathfrak{X}(\domi{u},\domi{p})}.\\
\meanfield{V \neq \mathfrak{f}}{\mathfrak{f} \mapsto \domi{b}} {\left(\bigwedge_i \psi_i\right)^\mathfrak{X}(\domi{u},\domi{p})} &= 
\prod_i\meanfield{V \neq \mathfrak{f}}{\mathfrak{f} \mapsto \domi{b}} {\psi_i^\mathfrak{X}(\domi{u},\domi{p})}.\\
%
%
\meanfield{V \neq \mathfrak{f}}{\mathfrak{f} \mapsto \domi{b}}{\lossfunc\left(\Auth, \mathfrak{X}; \varphi\right)} &=\displaystyle\sum_{(\domi{u}, \domi{p}) \in U \times P} \left|\mathit{Auth}(\domi{u}, \domi{p}) - \meanfield{V \neq \mathfrak{f}}{\mathfrak{f} \mapsto \domi{b}}{\varphi^\mathfrak{X}(\domi{u}, \domi{p})}\right|.
\end{align*}
\label{lem:expectations_flas}
\end{lemma}

Recall that $\land$ and $\neg$ form a complete set of Boolean operators. So one can also use this lemma to compute expectations of diverse random variables of the form $\left(\varphi \to \psi\right)^{\mathfrak{X}}(\domi{u}, \domi{p})$ and $\left(\varphi \lor \psi\right)^{\mathfrak{X}}(\domi{u}, \domi{p})$.

\section{RBAC mining with {\polymver}}
\label{sec:examples}

We explain next how to use {\polymver} to build an RBAC miner.

\subsection{RBAC policies}
\label{sub:rbac}

We already explained how the formula $\varphi^{\mathit{RBAC}}_N \in \Lang$ is a template formula for the language of all RBAC policies with at most $N$ roles. To implement Algorithm~\ref{algo:varinfer}, we only need a procedure to compute $\meanfield{V \neq \mathfrak{f}}{\mathfrak{f} \mapsto \domi{b}}{\lossfunc\left(\Auth, \mathfrak{X}; \varphi^\mathit{RBAC}_N\right)}$. Since, as noted in Example~\ref{ex:rbac_abac_diverse}, $\left(\varphi^\mathit{RBAC}_N\right)^{\mathfrak{X}}\left(\domi{u}, \domi{p}\right)$ is diverse, we can apply Lemma~\ref{lem:expectations_flas} to show that
\begin{align*}
&\meanfield{V \neq \mathfrak{f}}{\mathfrak{f} \mapsto \domi{b}}{\lossfunc\left(\Auth, \mathfrak{X}; \varphi^\mathit{RBAC}_N\right)} =\\
&\qquad \displaystyle\sum_{(\domi{u}, \domi{p}) \in U \times P} \left|\mathit{Auth}(\domi{u}, \domi{p}) - \meanfield{V \neq \mathfrak{f}}{\mathfrak{f} \mapsto \domi{b}}{\left(\varphi^\mathit{RBAC}_N\right)^\mathfrak{X}(\domi{u}, \domi{p})}\right|.
\end{align*}
\begin{align*}
&\meanfield{V \neq \mathfrak{f}}{\mathfrak{f} \mapsto \domi{b}}{\left(\varphi^\mathit{RBAC}_N\right)^\mathfrak{X}(\domi{u}, \domi{p})} =\\ 
&\qquad 1 - \prod_{i \leq N}\left(1 - \meanfield{V \neq \mathfrak{f}}{\mathfrak{f} \mapsto \domi{b}}{\mathit{UA}^\mathfrak{X}(\domi{u},\ricons{r}_i)}\meanfield{V \neq \mathfrak{f}}{\mathfrak{f} \mapsto \domi{b}}{\mathit{PA}^\mathfrak{X}(\ricons{r}_i,\domi{p})}\right),
\end{align*}
where, 
$$\meanfield{V \neq \mathfrak{f}}{\mathfrak{f} \mapsto \domi{b}}{\mathit{UA}^\mathfrak{X}(\domi{u},\ricons{r}_i)} = 
\begin{cases}
\domi{b} & \text{ if $\mathit{UA}^\mathfrak{X}(\domi{u},\ricons{r}_i) = \mathfrak{f}$}\\
\sum_{\domi{b}}\tilde{q}_{\mathit{UA}^\mathfrak{X}(\domi{u},\ricons{r}_i)}\left(\domi{b}\right) \domi{b} & \text{ otherwise. }
\end{cases}$$
$\meanfield{V \neq \mathfrak{f}}{\mathfrak{f} \mapsto \domi{b}}{\mathit{PA}^\mathfrak{X}(\ricons{r}_i,\domi{p})}$ is computed analogously. 

Observe how the computations of expectations is reduced to a simple rewriting procedure by applying Lemma~\ref{lem:expectations_flas}. We can now implement an RBAC miner by implementing Algorithm~\ref{algo:varinfer} in a standard programming language and using the results above to compute the needed expectations.

%
%


\subsection{Simple RBAC policies}
\label{sub:succinct_rbac}

The objective function used above has a limitation. When the number of role constants $N$ used by $\varphi^{\mathit{RBAC}}_N\left(\vari{u}, \vari{p}\right)$ is very large, we might obtain a policy $\tilde{\mathfrak{I}}$ that assigns each role to exactly one user. The role assigned to a user would be assigned all permissions that the user needs. As a result, $\lossfunc(\Auth, \tilde{\mathfrak{I}}; \varphi^{\mathit{RBAC}}_N) = 0$, but $\tilde{\mathfrak{I}}$ is not a desirable policy. We can avoid mining such policies by introducing in the objective function a \emph{regularization term} that measures the complexity of the mined policy $\mathfrak{I}$. A candidate regularization term is:
\begin{equation*}
\left\|\mathfrak{I}\right\| = \sum_{i \leq N} \left(\sum_{\domi{u} \in U} \mathit{UA}^\mathfrak{I}(\domi{u},\ricons{r}_i) + \sum_{\domi{p} \in P}\mathit{PA}^\mathfrak{I}(\ricons{r}_i,\domi{p})\right).
\label{eq:reg_rbac}
\end{equation*}
Observe that $\left\|\mathfrak{I}\right\|$ measures the sizes of the relations $\mathit{UA}^\mathfrak{I}$ and $\mathit{PA}^\mathfrak{I}$, for $i \leq N$, thereby providing a measure of $\mathfrak{I}$'s complexity. 
We now define the following loss function:
\begin{equation*}
\lossfunc^r_{\mathit{RBAC}}(\Auth, \mathfrak{I}) = \lambda \left\|\mathfrak{I}\right\| + \lossfunc(\Auth, \mathfrak{I}; \varphi).
\end{equation*}
Here $\lambda > 0$ is a trade-off hyper-parameter, which again must be fixed before executing the policy miner and can be estimated using grid search. Note that $\lossfunc^r_{\mathit{RBAC}}$ penalizes not only policies that substantially disagree with $\Auth$, but also policies that are too complex.

The computation of $\meanfield{V \neq \mathfrak{f}}{\mathfrak{f} \mapsto \domi{b}}{\lossfunc^r_{\mathit{RBAC}}(\Auth, \mathfrak{X})}$ now also requires computing $\meanfield{V \neq \mathfrak{f}}{\mathfrak{f} \mapsto \domi{b}}{\left\|\mathfrak{X}\right\|}$, where $\left\|\mathfrak{X}\right\|$ is the random variable obtained by replacing each occurrence of $\mathfrak{I}$ in $\left\|\mathfrak{I}\right\|$ with $\mathfrak{X}$. Fortunately, one can see that $\left\|\mathfrak{X}\right\|$ is diverse. Hence, we can use the linearity of expectation and Lemma~\ref{lem:expectations_flas} to compute all needed expectations.

\section{Mining spatio-temporal RBAC policies}
\label{sub:st_rbac}

\newcommand{\timeSort}{\sortFont{TIME}}
\newcommand{\spaceSort}{\sortFont{SPACE}}
\newcommand{\riconsts}[1]{C\left(#1\right)}

We now use {\polymver} to build the first policy miner for RBAC extensions with spatio-temporal constraints~\cite{aich2007starbac, cotrini2015analyzing, joshi2004access, bhatti2005x, ben2016gemrbac, ben2016model}. In policies in these extensions, users are assigned permissions not only according to their roles, but also based on constraints depending on the current time and the user's and the permission's locations. The syntax for specifying these constraints allows for policies like ``a user is assigned the role Engineer from Monday through Friday and from 8:00 AM until 5:00 PM'' or ``the role Engineer is granted permission to access any object within a radius of three miles from the main building.''

We present a template formula $\varphi^{\mathit{st}}(\vari{t}, \vari{u}, \vari{p}) \in \Lang$ for a policy language that we call \emph{spatio-temporal RBAC}. This is an extension of RBAC with a syntax for spatial constraints based on~\cite{ben2016gemrbac, ben2016model} and a syntax for temporal constraints based on temporal RBAC~\cite{bertino2001trbac}.
$$\varphi^{\mathit{st}}(\vari{t}, \vari{u}, \vari{p}) = \bigvee_{i \leq \NumRoles} \left(\psi_\mathit{UA}(\vari{t}, \vari{u}, \ricons{r}_i) \land \psi_\mathit{PA}(\vari{t}, \ricons{r}_i, \vari{p})\right).$$
Here, we assume the existence of a sort $\timeSort$ and that $\vari{t}$ is a variable of this sort representing the time when $\vari{u}$ exercises $\vari{p}$. We also assume the existence of a sort $\spaceSort$ that we use to specify spatial constraints. The formulas $\psi_\mathit{UA}(\vari{t}, \vari{u}, \ricons{r}_i)$ and $\psi_\mathit{PA}(\vari{t}, \ricons{r}_i, \vari{p})$ describe when a user is assigned the role $\ricons{r}_i$ and when a permission is assigned to the role $\ricons{r}_i$, respectively. We use the rigid constants $\ricons{r}_1, \ldots, \ricons{r}_\NumRoles$ to denote roles.

The grammar $\Gamma_{\textit{st}}$ below defines the syntax of $\psi_{\mathit{UA}}$ and $\psi_{\mathit{PA}}$.
\newcommand{\gramterm}[1]{\left\langle \texttt{#1}\right\rangle}
\vspace{-1pt}
\begin{align*}
\gramterm{cstr\_list} &::=  \gramterm{cstr} ( \; \land \gramterm{cstr} \; )*\\
\gramterm{cstr} &::= \gramterm{sp\_cstr} (\; \lor \; \gramterm{sp\_cstr} \;)* \mid\\
&\phantom{::= }\gramterm{tmp\_cstr} (\; \lor \; \gramterm{tmp\_cstr} \;)*\\
\gramterm{sp\_cstr} &::= 
	 \left(\neg \texttt{?}\right) \underline{\mathit{isWithin}}\left(\underline{\mathit{Loc}}\left(\vari{o}\right), \cons{d}, \cons{b}\right)\\
\gramterm{tmp\_cstr} &::= \psi_{\mathit{cal}}\left(\vari{t}\right)
\end{align*}


An expression in this grammar is a conjunction of \emph{constraints}, each of which is either a disjunction of \emph{temporal constraints} or a disjunction of \emph{spatial constraints}. 

\subsection{Modeling spatial constraints}
\label{sub:uni_modeling_spatial_constrains}

A \emph{spatial constraint} is a (possibly negated) formula of the form $\underline
{\mathit{isWithin}}\left(\underline{\mathit{Loc}}\left(\vari{o}\right), \cons{d}, \cons{b}\right)$, where $\vari{o}$ is a variable of sort $\Users$ or $\Perms$, $\underline{\mathit{Loc}}\left(\vari{o}\right)$ denotes $\vari{o}$'s location, $\cons{d}$ is a flexible constant symbol of a sort whose carrier set is $\mathbb{N}^{\leq M} = \left\{\cons{0}, \cons{1}, \ldots, \cons{M}\right\}$ (where $M$ is a value fixed in advance), and $\cons{b}$ is a flexible constant symbol of a sort describing the organization's physical facilities. For example, $\underline{\mathit{isWithin}}\left(\underline{\mathit{Loc}}\left(\vari{u}\right), \cons{4}, \cons{MainBuilding}\right)$ holds when the user represented by $\vari{u}$ is within 4 space units of the main building.

Intuitively, the formula $\underline{\mathit{isWithin}}\left(\underline{\mathit{Loc}}\left(\vari{o}\right), \cons{d}, \cons{b}\right)$ evaluates whether the entity represented by $\vari{o}$ is located within $\cons{d}$ spatial units from $\cons{b}$. Observe that a policy miner does not need to compute interpretations for rigid function symbols like $\underline{\mathit{Loc}}$ or rigid relation symbols like $\underline{\mathit{isWithin}}$, since they already have a fixed interpretation.\looseness=-1

\subsection{Modeling temporal constraints}
\label{sub:uni_modeling_temporal_constraints}

A \emph{temporal constraint} is a formula $\psi_{\mathit{cal}}\left(\vari{t}\right)$ that represents a \emph{periodic expression}~\cite{bertino2001trbac}, which describes a set of time intervals. We give here a simplified overview and refer to the literature for full details~\cite{bertino2001trbac}.

\begin{definition}
A \emph{periodic expression} is a tuple $(\mathit{yearSet}, \mathit{monthSet}$, $\mathit{daySet}, \mathit{hourSet}, \mathit{hourDuration}) \in \left(2^\mathbb{N}\right)^4 \times \mathbb{N}.$
A \emph{time instant} is a tuple $(y, m, d, h) \in \mathbb{N}^4$. The time instant satisfies the periodic expression
if $y \in \mathit{yearSet}$, $m \in \mathit{monthSet}$, $d \in \mathit{daySet}$, and there is an $h' \in \mathit{hourSet}$ such that $h' \leq h \leq h' + \mathit{hourDuration}$.
\end{definition}

Previous works on analyzing temporal RBAC with SMT solvers~\cite{jha2014security} show that temporal constraints can be expressed as formulas in $\Lang$. Furthermore, one can verify that any expression in $\Gamma_{\textit{st}}$ and, therefore, $\varphi^{\mathit{st}}$ is in $\Lang$.\looseness=-1

As an objective function, we use $\lambda \left\|\mathfrak{I}\right\| + \lossfunc\left(\Auth, \mathfrak{I}; \varphi^{\mathit{st}}\right)$. Here, $\left\|\mathfrak{I}\right\|$ counts the number of spatial constraints plus the sum of the \emph{weighted structural complexities} of all temporal constraints~\cite{stoller2016mining}. For computing expectations, one can show that $\left\|\mathfrak{X}\right\|$ is diverse and that every atomic formula in $\varphi^{\mathit{st}}$ occurs exactly once. Hence, one can compute all necessary expectations using the linearity of expectation and Lemma~\ref{lem:expectations_flas}.\looseness=-1

\section{Experiments}
\label{sec:experiments}

In this section, we experimentally validate two hypotheses. First, using {\polymver}, we can build policy miners for a wide variety of policy languages. Second, the policies mined by these miners have as low complexity and high generalization ability as those mined by the state of the art.

\subsection{Datasets}
\label{sub:datasets}

Our experiments are divided into the following categories.

\paragraph{Mine RBAC policies from access control matrices} We use three access control matrices from three real organizations, named ``healthcare'', ``firewall'', and ``americas''~\cite{ene2008fast}. For healthcare, there are 46 users and 46 permissions, for firewall, there are 720 users and 587 permissions, and for americas, there are more than 10,000 users and around 3,500 permissions. We refer to these access control matrices as RBAC1, RBAC2, and RBAC3.

\paragraph{Mine ABAC policies from logs} We use four logs of access requests provided by Amazon for a Kaggle competition in 2013~\cite{amazonchallenge}, where participants had to develop mining algorithms that predicted from the logs which permissions must be assigned to which users. We refer to these logs as ABAC1, ABAC2, ABAC3, and ABAC4.

\paragraph{Mine business-meaningful RBAC policies from access control matrices} We use the access control matrix provided by Amazon for the IEEE MLSP 2012 competition~\cite{mlsp2012amazon}, available at the UCI machine learning repository~\cite{amazonuci}. It assigns three types of permissions, named ``HOST'', ``PERM\_GROUP'', and ``SYSTEM\_GROUP'' to 30,000 users. The number of permissions for each type are approximately 1,700, 6,000, and 20,000, respectively. For each type of permission, we sampled 5,000 users from all 30,000 users and used all permissions of that type to build an access control matrix. We explain in detail how we create these matrices in Appendix~\ref{sub:bm_rbac_datasets}. We refer to these matrices as BM-RBAC1, BM-RBAC2, and BM-RBAC3.

\paragraph{Mine XACML policies from access control matrices} We use Continue~\cite{krishnamurthi2003continue, fisler2005verification}, the most complex set of XACML policies in the literature. We use seven of the largest policies in the set. For each of them, we compute the set of all possible requests and decide which of them are authorized by the policy. We then mine a policy from this set of decided requests. For the simplest policy, there are around 60 requests and for the most complex policy, there are more than 30,000 requests. We call these seven sets of requests XACML1, XACML2,~..., XACML7.\looseness=-1

\paragraph{Mine spatio-temporal RBAC policies from logs} There are no publicly available datasets for mining spatio-temporal RBAC policies. Based on policies provided as examples in recent works~\cite{ben2016gemrbac, ben2016model}, we created a synthetic policy and a synthetic log by creating 1,000 access requests uniformly at random and evaluating them against the policy. We refer to this log as STARBAC. The synthetic policy is described in Appendix~\ref{sub:synth_pol_starbac}.

\subsection{Methodology}
\label{sub:uni_methodology}

For RBAC and ABAC, we mine two policies in the corresponding policy language's syntax. The first one using a miner built using {\polymver} and the second one using a state-of-the-art miner. Details on the miners built using {\polymver} are given in Sections~\ref{sec:examples} and~\ref{sub:st_rbac}, and Appendices~\ref{sec:app_examples} and~\ref{sub:xacml}. For RBAC, we use for comparison the miner presented in~\cite{frank2013role} and, for ABAC, we use for comparison the miner from~\cite{cotrini2018rhapsody}. For XACML and spatio-temporal RBAC, there are no other known miners. For business meaningful RBAC, we contacted the authors of miners for this RBAC extension~\cite{frank2013role, molloy2012generative}, but implementations of their algorithms were not available.


As an objective function we use $\lambda \left\|\mathfrak{I}\right\| + \lossfunc\left(\Auth, \mathfrak{I}; \varphi\right)$, where $\lambda$ is a trade-off hyper-parameter, $\left\|\mathfrak{I}\right\|$ is the complexity measure defined for $\mathfrak{I}$ in the policy language, and $\varphi$ is the template formula for the corresponding policy language. The values for the hyper-parameters were computed using grid search.

To evaluate miners for RBAC, BM-RBAC, and XACML, we use 5-fold cross-validation~\cite{d2015solution, yang2001mining, de2013data}, as described in Appendix~\ref{sub:evaluation}. To measure the mined policy's generalizability, we measure its \emph{true positive rate} (TPR) and its \emph{false positive rate} (FPR)~\cite{powers2011evaluation}. To measure a mined policy's complexity, we use $\left\|\mathfrak{I}\right\|$. To evaluate miners for ABAC and STARBAC, which receive a log instead of an access control matrix as input, we use universal cross-validation~\cite{cotrini2018rhapsody}. We measure the mined policy's TPR, FPR, precision, and complexity. We considered only those mined policies whose FPR was below 5\%.

All policy miners, except the one for BM-RBAC, were developed in Python 3.6 and were executed on machines with 2,8 GHz 8-core CPUs and 32 GB of RAM. The miner for BM-RBAC was developed in Pytorch version 0.4~\cite{paszke2017automatic} and executed on an NVIDIA GTX Titan X GPU with 12 GB of RAM. 
For all policy languages except STARBAC, our experiments finished within 4 hours. For STARBAC, they took 7 hours. We remark that organizations do not need to mine policies on a regular basis, so policies need not be mined in real time~\cite{cotrini2018rhapsody}.\looseness=-1

\subsection{Results}
\label{sub:results}

Figures~\ref{fig:tprs}--\ref{fig:precs} compare, respectively, the TPRs, complexities, and precisions of the policies we mined with those mined by the state of the art across the different datasets with respect to the different policy languages. We make the following observations.

%
%
%

\begin{itemize}
\item We mine policies whose TPR is within 5\% of the state-of-the-art policies' TPR. For the XACML and STARBAC scenarios, where no other miners exist, we mine policies with a TPR above 75\% in all cases.
\item In most cases, we mine policies with a complexity lower than the complexity of policies mined by the state of the art.
\item When mining from logs, we mine policies that have a similar or greater precision than those mined by the state of the art, sometimes substantially greater.
\item In all cases, we mine policies with an FPR $\leq 5\%$ (not shown in the figures).
\end{itemize} 

\begin{figure}
\centering
\includegraphics[scale=1]{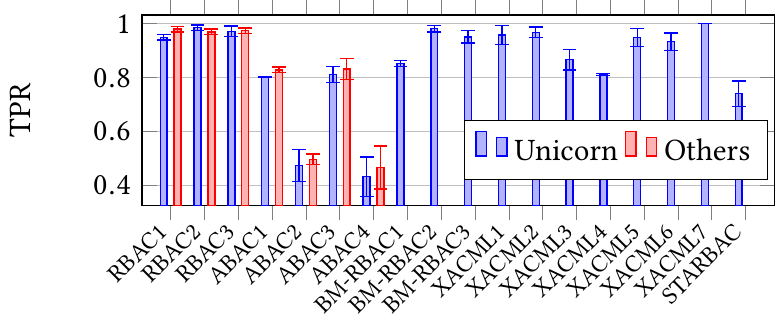} 
\caption{Comparison of the TPRs between policies mined using {\polymver} and policies mined by the state of the art across different policy languages. Policies with \emph{higher} TPRs are better at granting permissions to the correct users.}
\label{fig:tprs}
\end{figure}

\begin{figure}
\centering
\includegraphics[scale=1]{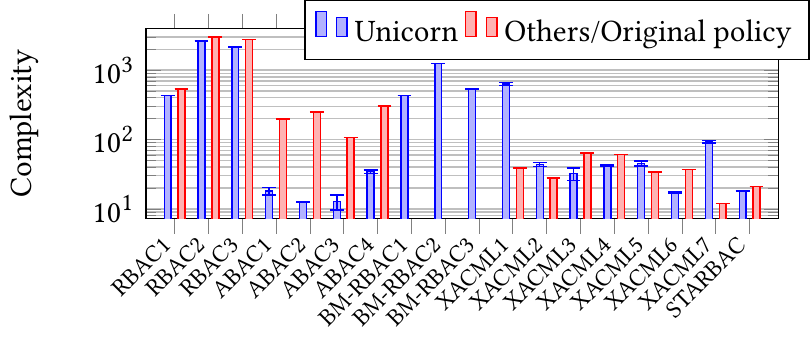} 
%
\caption{Comparison of the complexities between policies mined using {\polymver} and policies mined using the state of the art across different policy languages. Policies with \emph{lower} complexities are better as they are easier to interpret by humans. For XACML and STARBAC, there is no known miner, but we compared the mined policy's complexity with that of the original policy.}
\label{fig:sizes}
\end{figure}

\begin{figure}
\centering
\includegraphics[scale=0.9]{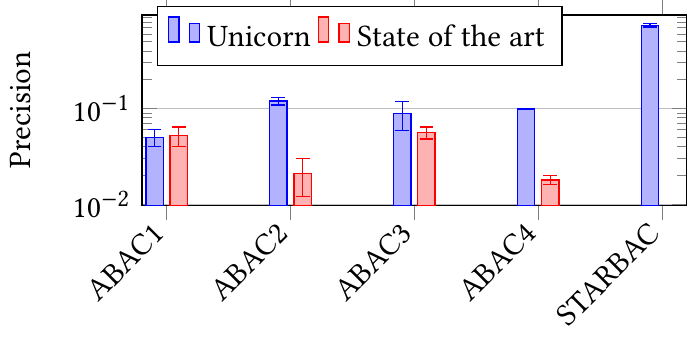} 
\caption{Comparison of the precision between policies mined using {\polymver} and policies mined by a state of the art policy miner across different policy languages. Policies with \emph{higher} precision are better as they avoid incorrect authorizations. We only compare the precision of mined policies when mining from logs, as discussed in~\cite{cotrini2018rhapsody}.}
\label{fig:precs}
\end{figure}

\subsection{Discussion}
\label{sub:discussion}

Our experimental results show that, with the exception of ABAC, all policies we mined attain a TPR of at least 80\% in most of the cases. The low TPR in ABAC is due to the fact that the logs contain only 7\% of all possible requests~\cite{cotrini2018rhapsody}. But even in that case, the ABAC miner we built attains a TPR that is within 5\% of the TPR attained by the state of the art~\cite{cotrini2018rhapsody}. Moreover, our ABAC miner mines policies with substantially lower complexity and higher precision. 
These results support our hypothesis that by using {\polymver} we can build competitive policy miners for a wide variety of policy languages.

These results also suggest that the miners built are well-suited for practical use. In this regard, note that policy miners are tools that facilitate the specification and maintenance of policies. They are not intended to replace human policy administrators, especially when the miners work on logs. This is because logs contain just an incomplete view of how permissions should be assigned to users. Very sparse logs, like those used for the experiments on ABAC, contain barely 7\% of all possible authorization requests. Hence, we cannot expect policy miners to deduce how all permissions should be assigned from such logs. The policy administrator must review the mined policy and specify how it should decide groups of requests that are not well represented in the log. For this reason, mined policies must also be simple. The main application of policy miners is to reduce the cumbersome effort of manually analyzing logs (or, more generally, permission assignments) and mine policies that \emph{generalize well} (see Section~\ref{sub:quality_criteria}).\looseness=-1

Observe that the mined policies correctly authorize at least 40\% of future requests in all cases for ABAC and that in some cases they correctly authorize 80\% of all requests. All this with a false positive rate below 5\%. This means that the mined policy has already reduced the policy administrator's work by at least 40\% and in most of the cases by at least 80\%. The administrator now only needs to decide how the policy should decide groups of requests that are not represented in the log.

\section{Related work}
\label{sec:related_work}


\subsection{Policy mining}
\label{sec:policy_mining}

\subsubsection{RBAC mining}

Early research on policy mining focused on RBAC~\cite{ene2008fast, kuhlmann2003role, vaidya2010role}. The approaches developed used combinatorial algorithms to find, for an assignment of permissions to users, an approximately minimal set of role assignments, e.g.,~\cite{vaidya2007role, lu2008optimal, schlegelmilch2005role, vaidya2006roleminer, zhang2007role}. A major step forward was the use of machine-learning techniques like latent Dirichlet allocation~\cite{molloy2012generative} and deterministic annealing~\cite{frank2013role, streich2009multi} to compute models that maximize the likelihood of the given assignment of permission to users. More recent works mine RBAC policies with time constraints~\cite{mitra2015generalized, mitra2016mining} and role hierarchies~\cite{stoller2016mining, guo2008role}, using combinatorial techniques that are specific to the RBAC extension.\looseness=-1

Despite the plethora of RBAC miners, there are still many RBAC extensions for which no miner has been developed. A recent survey in role mining~\cite{mitra2016survey}, covering over a dozen RBAC miners, reports not a single RBAC miner that can mine spatio-temporal constraints, even though there have been several spatio-temporal extensions of RBAC since 2000, e.g.,~\cite{ray2007spatio, chen2008spatio, kumar2006strbac, toahchoodee2009ensuring, chandran2005lot, aich2007starbac, cui2007ex}, and additional extensions are under way~\cite{ben2016gemrbac, ben2016model}. {\polymver} offers a practical solution to mining RBAC policies for these extensions. As illustrated in Section~\ref{sub:st_rbac}, we can now mine spatio-temporal RBAC policies.\looseness=-1

\subsubsection{Other miners}

Miners have recently been proposed for other policy languages like ABAC~\cite{xu2014mining, cotrini2018rhapsody} and ReBAC (Relationship-Based Access Control)~\cite{bui19mining}. These algorithms use dedicated combinatorial and machine-learning methods to mine policies tailored to the given policy language. {\polymver} has the advantage of being applicable to a much broader class of policy languages.

%


\subsection{Interpretable machine learning}
\label{sec:interpretable_ml}

Machine-learning algorithms have been proposed that train an \emph{interpretable} model~\cite{kavvsek2006apriori, jovanoski2001classification, steinberg2009cart, clark1991rule, angelino2017learning} consisting of a set of human-readable rules that describe how an instance is classified. Such algorithms are attractive for policy mining, as policies must not only correctly grant and deny access, they should also be easy to understand.\looseness=-1

The main limitation of the rules mined by these models is that they often do not comply with the underlying policy language's syntax. State-of-the-art algorithms in this field~\cite{steinberg2009cart, clark1991rule, angelino2017learning} produce rules that are simply conjunctions of constraints on the instances' features. This is insufficient for many policy languages, like XACML, where policies can consist of nested subpolicies that are composed with XACML's policy combination algorithms~\cite{godik2002oasis}. 

The main advantage of {\polymver} is that it can mine policies that not only correctly grant and deny access in most cases, but are also compliant with a given policy language's syntax, like XACML. Moreover, as illustrated in Section~\ref{sub:succinct_rbac}, one can tailor the objective function so that the policy miner searches for a simple policy.

\section{Conclusion}
\label{sec:uni_conclusion}

The difficulty of specifying and maintaining access control policies has spawned a large and growing number of policy languages with associated policy miners. However, developing such miners is challenging and substantially more difficult than creating a new policy language. This problem is exacerbated by the fact that existing mining algorithms are inflexible in that they cannot be easily modified to mine policies for other policy languages with different features. In this paper, we demonstrated that it is in fact possible to create a universal method for building policy miners that works very well for a wide variety of policy languages.

We validated {\polymver}'s effectiveness experimentally, including a comparison against state-of-the-art policy miners for different policy languages. In all cases, the miners built using {\polymver} are competitive with the state of the art.

As future work, we plan to automate completely the workflow in Figure~\ref{fig:unicorn}. We envision a \emph{universal policy mining algorithm} based on Algorithm~\ref{algo:varinfer} that, given as input the policy language, the permission assignment, and the objective function, automatically computes the probabilistic model and the most likely policy constrained by the given permission assignment.

%% file: unicorn_appendix.tex
\section{Proof of Lemma~\ref{lem:cond_prob_is_vec_prob}}
\label{sec:proof_cond_prob_is_vec_prob}

\begin{customthm}{\ref{lem:cond_prob_is_vec_prob}} For a policy language $\varphi \in \Lang$,
\begin{equation}
\prob\left(\mathfrak{I} \mid \Auth\right) = \prob\left(\left(\mathfrak{f}^{\mathfrak{X}}\right)_{\mathfrak{f} \in \atomfacts{\varphi}{}} = \left(\mathfrak{f}^{\mathfrak{I}}\right)_{\mathfrak{f} \in \atomfacts{\varphi}{}} \middle| \Auth\right).
\end{equation}
%
\end{customthm}

\begin{proof}
The first equality follows from $\mathfrak{X}$'s definition, so we prove the second equality only. 
Let $\mathcal{F}$ be the random vector $\left(\mathfrak{f}^{\mathfrak{X}}\right)_{\mathfrak{f} \in \atomfacts{\varphi}{}}$. That is, $\mathcal{F}$ is a random variable that maps $\mathfrak{I}' \in \Omega$ to $\left(\mathfrak{f}^{\mathfrak{I}'}\right)_{\mathfrak{f} \in \atomfacts{\varphi}{}}$. It suffices to show that $E_1 := \left\{\mathfrak{I}\right\}$ and $E_2 := \left\{\mathfrak{I}' \mid \mathcal{F}\left(\mathfrak{I}'\right) = \left(\mathfrak{f}^{\mathfrak{I}}\right)_{\mathfrak{f} \in \atomfacts{\varphi}{}}\right\}$ are the same. To achieve this, it suffices to show that, for any $\mathfrak{I}'$, if $\mathcal{F}\left(\mathfrak{I}'\right) = \left(\mathfrak{f}^{\mathfrak{I}}\right)_{\mathfrak{f} \in \atomfacts{\varphi}{}}$ then $\mathfrak{I}' = \mathfrak{I}$. But this is true, as the values $\left(\mathfrak{f}^{\mathfrak{I}}\right)_{\mathfrak{f} \in \atomfacts{\varphi}{}}$ also turn out to completely determine $\mathfrak{I}$. We conclude then that $E_1 = E_2$, which yields the desired result.
\end{proof}

\section{Evaluating policy miners}
\label{sub:evaluation}

We describe here evaluation methods for policy miners. Recall from Section~\ref{sub:quality_criteria} that policy miners are evaluated by the generalization and the complexity of the policies they mine. 

For the sake of clarity, we assume that a policy miner $M$ is a program that receives as input a subset $A \subseteq U \times P$ denoting the currently known permission assignment. 

\subsection{Cross-validation}
\label{sub:cross_val}

To measure the complexity of a mined policy, one simply evaluates a suitable metric on the policy. Measuring a mined policy's generalization is harder and is related to a fundamental question in machine learning: how to measure a model's prediction accuracy on new data. The generalization ability can be measured with a technique called \emph{cross-validation}~\cite{friedman2001elements}. We give an informal overview and point to the references for a formal explanation~\cite{bishop2006pattern, friedman2001elements}. In particular, for cross-validation in RBAC, we refer to~\cite{frank2013role}, and for cross-validation in ABAC with logs, we refer to~\cite{cotrini2018rhapsody}.

In cross-validation, we take a representative set $Q \subseteq U \times P$ of all possible requests. Recall that a request is a pair consisting of a user and a permission. We assume known a relation $\mathit{Auth} \subseteq U \times P$ that specifies which requests should be authorized. The set $Q$ is then split uniformly at random into two disjoint sets: a training set $\mathit{Tr}$ and a testing set $\mathit{Ts}$. $\mathit{Auth} \cap \mathit{Tr}$ is given as input to the policy miner. Let $\pi$ be the mined policy. Let $A := \mathit{Ts} \cap \mathit{Auth}$ and $D := \mathit{Ts} \setminus A$. Let $\hat{A} := \{q \in \mathit{Ts} \mid \pi(q) = 1\}$ be the requests authorized by the policy and let $\hat{D} := \mathit{Ts} \setminus \hat{A}$. We measure the mined policy's generalizability with the \emph{true positive rate} (TPR) and the \emph{false positive rate} (FPR)~\cite{powers2011evaluation}:
\begin{equation*}
\mathit{TPR} := \frac{\left|\hat{A} \cap A\right|}{\left|A\right|} \qquad \mathit{FPR} := \frac{\left|\hat{A} \cap D\right|}{\left|D\right|}.
\label{eq:tpr_fpr}
\end{equation*}
Intuitively, since the TPR and the FPR are computed on requests that were not part of the miner's input, they evaluate how well the mined policy decides new requests in the future. Hence, good policies have a high TPR and a low FPR.

\paragraph{$K$-fold cross-validation} This is a variant of cross-validation. Instead of splitting $Q$ into two disjoint sets, $Q$ is split uniformly at random into $K$ sets $Q_1, \ldots, Q_K$. Afterwards, for each $i \leq K$, one performs cross-validation using $\cup_{j \neq i} Q_j$ as a training set and $Q_i$ as a testing set. Each of these $K$ iterations is called a \emph{fold}. Finally, one measures the average of the true positive rates and the average of the false positive rates over all $K$ folds. 

\subsection{Grid search for policy mining}
\label{sub:grid_search}

Policy miners sometimes require input values for \emph{hyper-parameters}. To determine the values that make the miners compute the best policies, we use \emph{grid search}~\cite{hyperparamsearch}, which we briefly recall next. 

Let $\alpha_1, \ldots, \alpha_N$ be the hyper-parameters of a policy miner. In grid search, for $i \leq N$, one defines a set $D_i$ of \emph{candidate values}. Then for each $(a_1, \ldots, a_N) \in D_1 \times \ldots \times D_N$, one executes cross-validation using $(a_1, \ldots, a_N)$ as values for the hyper-parameters. Let $\text{TPR}\left(a_1, \ldots, a_N\right)$ and $\text{FPR}\left(a_1, \ldots, a_N\right)$ be the TPR and FPR of the policy mined during that execution of cross-validation, respectively. Finally, one chooses the tuple $(a^*_1, \ldots, a^*_N)$ that maximizes $\text{TPR}\left(a_1, \ldots, a_N\right)$ subject to $\text{FPR}\left(a^*_1, \ldots, a^*_N\right) \leq c$.
%
%
The threshold $c$ for the FPR is arbitrary (we used 0.05) and can be adjusted. It defines a maximum bound of false positives that can be tolerated from a mined policy.

\section{Simplifying the computation of expectations}
\label{sec:app_uni_expectation_computation}

We prove here Lemma~\ref{lem:expectations_flas}. We start with some auxiliary lemmas and definitions.

\begin{lemma} Let $\mathfrak{f}$ and $\mathfrak{g}$ be facts, $\varphi$ be a formula in $\Lang$, $(\domi{u},\domi{p}) \in U \times P$, and 
 $\{\psi_i\}_i \subseteq \Lang$ such that $\{\psi_i^{\mathfrak{X}}(\domi{u},\domi{p})\}_i$ is a set of mutually independent random variables under the distribution $q$.
\begin{align*}
\meanfield{V \neq \mathfrak{f}}{\mathfrak{f} \mapsto \domi{b}} {\mathfrak{g}} &= 
\begin{cases}
\domi{b} & \text{ if $\mathfrak{f} = \mathfrak{g}$}\\
\sum_{\domi{b} \in \rangerandfact{\mathfrak{g}}}\tilde{q}_{\mathfrak{g}}\left(\domi{b}\right) \domi{b} & \text{ otherwise. }
\end{cases}\\
\meanfield{V \neq \mathfrak{f}}{\mathfrak{f} \mapsto \domi{b}} {\left(\neg \varphi\right)^\mathfrak{X}(\domi{u},\domi{p})} &= 1 - \meanfield{V \neq \mathfrak{f}}{\mathfrak{f} \mapsto \domi{b}} {\varphi^\mathfrak{X}(\domi{u},\domi{p})}.\\
\meanfield{V \neq \mathfrak{f}}{\mathfrak{f} \mapsto \domi{b}} {\left(\bigwedge_i \psi_i\right)^\mathfrak{X}(\domi{u},\domi{p})} &= 
\prod_i\meanfield{V \neq \mathfrak{f}}{\mathfrak{f} \mapsto \domi{b}} {\psi_i^\mathfrak{X}(\domi{u},\domi{p})}.\\
%
%
\end{align*}
\label{lem:pre_expectations_flas}
\end{lemma}

\begin{proof}
Observe that, for a Bernoulli random variable $X$, $\mathbb{E}\left[X\right] = \prob\left(X = 1\right)$. Recall also that $\mathbb{E}\left[XY\right] = \mathbb{E}\left[X\right]\mathbb{E}\left[Y\right]$, whenever $X$ and $Y$ are mutually independent. With these observations and using standard probability laws, one can derive the equations above.
\end{proof}
%

\begin{lemma}
Let $\varphi \in \Lang$ and let $(\domi{u}, \domi{p}) \in U \times P$. If $\varphi^{\mathfrak{X}}\left(\domi{u}, \domi{p}\right)$ is diverse, then $\meanfield{}{\mathfrak{f} \mapsto \domi{b}}{\varphi^\mathfrak{X}\left(\domi{u}, \domi{p}\right)}$ can be computed using only the equations from Lemma~\ref{lem:expectations_flas}.
\label{lem:diverse_expectations}
\end{lemma}

This lemma is proved by induction on $\varphi$ and by recalling that any two different random facts are independent under the distribution $q$, which follows from the way that the distribution $q$ is factorized.

%


\begin{corollary}
\begin{equation*}
\meanfield{V \neq \mathfrak{f}}{\mathfrak{f} \mapsto \domi{b}}{\lossfunc\left(\Auth, \mathfrak{X}; \varphi\right)} =\sum_{(\domi{u}, \domi{p}) \in U \times P} \left|\mathit{Auth}(\domi{u}, \domi{p}) - \meanfield{V \neq \mathfrak{f}}{\mathfrak{f} \mapsto \domi{b}}{\varphi^\mathfrak{X}(\domi{u}, \domi{p})}\right|.
\end{equation*}
\label{cor:expectation_loss_func}
\end{corollary}

\begin{proof}
$\lossfunc\left(\Auth, \mathfrak{X}; \varphi\right)$ can be rewritten as follows:
$$
\sum_{(\domi{u}, \domi{p}) \in \Auth} \left(1 - \varphi^\mathfrak{X}\left(\domi{u}, \domi{p}\right)\right) + \sum_{(\domi{u}, \domi{p}) \in U \times P \setminus \Auth} \varphi^{\mathfrak{X}}(\domi{u}, \domi{p}). 
$$
The result follows from the linearity of expectation.
\end{proof}

Lemma~\ref{lem:expectations_flas} follows from Lemma~\ref{lem:diverse_expectations} and Corollary~\ref{cor:expectation_loss_func}.

\section{Policy miners built using Unicorn}
\label{sec:app_examples}

We present here technical details on how we built policy miners for different policy languages using {\polymver}.

\subsection{ABAC policies}
\label{sub:abac}

\newcommand{\SortAttVals}{\mathit{AttVals}}
\newcommand{\SortRules}{\mathit{Rules}}
\newcommand{\NumAttVals}{M}
\newcommand{\NumRules}{N}
ABAC is an access control paradigm where permissions are assigned to users depending on the users' and the permission's \emph{attribute values}. An \emph{ABAC policy} is a set of \emph{rules}. A rule is a set of attribute values. Recall that a request $(u,p)$ is a pair consisting of a user $u \in U$ and a permission $p \in P$. A rule \emph{assigns} a permission $p$ to a user $u$ if $u$ and $p$'s permission attribute values contain all of the rule's attribute values. A policy \emph{assigns} $p$ to $u$ if some rule in the policy assigns $p$ to $u$.

When mining ABAC policies, we are not only given a permission assignment $\Auth \subseteq U \times P$, but also \emph{attribute assignment relations} $\mathit{UAtt} \subseteq U \times \mathit{AttVals}$ and $\mathit{PAtt} \subseteq P \times \mathit{AttVals}$ that describe what attribute values each user and each permission has. Here, $\mathit{AttVals}$ denotes the set of possible attribute values. We refer to previous work for a discussion on how to obtain these attribute assignment relations~\cite{xu2015mining, cotrini2018rhapsody}.\looseness=-1

The objective in mining ABAC policies is to find a set of rules that assigns permissions to users based on the users' and the permissions' attribute values. We explain next how to build a policy miner for ABAC using {\polymver}.
Let $\SortRules$ and $\SortAttVals$ be sorts for rules and attribute values, respectively. Let $\mathit{RUA}$ and $\mathit{RPA}$ be flexible binary relation symbols of type $\SortRules \times \SortAttVals$. For $\NumAttVals, \NumRules \in \mathbb{N}$, the formula $\varphi^{\mathit{ABAC}}_{\NumAttVals, \NumRules}(\vari{u}, \vari{p})$ below is a template formula for ABAC:
\begin{equation}
\bigvee_{i \leq \NumRules} \bigwedge_{j \leq {\NumAttVals}} 
\left(
\begin{array}{l}
\left(\mathit{RUA}\left(\ricons{s}_i, \ricons{a}_j\right) \to \underline{\mathit{UAtt}}(\vari{u}, \ricons{a}_j)\right) \land\\
\left(\mathit{RPA}\left(\ricons{s}_i, \ricons{a}_j\right) \to \underline{\mathit{PAtt}}(\vari{p}, \ricons{a}_j)\right)
\end{array}
\right).
\label{fla:abac_auth}
\end{equation}

In this formula, $\ricons{s}_i$, for $i \leq \NumRules$, is a rigid constant symbol of sort $\SortRules$ denoting a rule. The symbol $\ricons{a}_j$, for $j \leq \NumAttVals$, is a rigid constant denoting an attribute value. The formula $\mathit{RUA}\left(\ricons{s}_i, \ricons{a}_j\right)$ describes whether rule $\ricons{s}_i$ requires the user to have the attribute value $\ricons{a}_j$. The formula $\mathit{RPA}\left(\ricons{s}_i, \ricons{a}_j\right)$ describes an analogous requirement. We use two rigid relation symbols $\underline{\mathit{UAtt}}$ and $\underline{\mathit{PAtt}}$ to represent the attribute assignment relations. The formulas $\underline{\mathit{UAtt}}\left(\vari{u}, \ricons{a}_j\right)$ and $\underline{\mathit{PAtt}}\left(\vari{p}, \ricons{a}_j\right)$ describe whether $\vari{u}$ and $\vari{p}$, respectively, are assigned the attribute value $\ricons{a}_j$. Intuitively, the formula $\varphi^{\mathit{ABAC}}_{\NumAttVals, \NumRules}(\vari{u}, \vari{p})$ is satisfied by $(\domi{u}, \domi{p}) \in U \times P$ if, for some rule $\ricons{s}_i$, $(\domi{u}, \domi{p})$ possesses all user and permission attribute values required by $\ricons{s}_i$ under $\mathit{RUA}$ and $\mathit{RPA}$.\looseness=-1

Observe that a policy miner does not need to find an interpretation for the symbols $\underline{\mathit{UAtt}}$ and $\underline{\mathit{PAtt}}$ because \emph{the organization already has interpretations for those symbols}. When mining ABAC policies, the organization already knows what attribute values each user and each permission has and wants to mine from them an ABAC policy. The miner only needs to specify which attribute values must be required by each rule. This is why we specify the attribute assignment relations with rigid symbols.\looseness=-1

We use $\lossfunc(\Auth, \mathfrak{I}; \varphi^{\mathit{ABAC}}_{\NumAttVals, \NumRules})$ as the objective function. Observe that every atomic formula occurs at most once in $\varphi^{\mathit{ABAC}}_{\NumAttVals, \NumRules}$, so, by Corollary~\ref{cor:occurs_once}, we can use Lemma~\ref{lem:expectations_flas} to compute all relevant expectations.

Finally, we can also add a regularization term to $\lossfunc(\Auth, \mathfrak{I}; \varphi^{\mathit{ABAC}}_{\NumAttVals, \NumRules})$ to avoid mining policies with too many rules or unnecessarily large rules. One such regularization term is\looseness=-1
$$\left\|\mathfrak{I}\right\| = \sum_{i \leq \NumRules}\sum_{j \leq \NumAttVals} \mathit{RUA}^\mathfrak{I}\left(\ricons{s}_i, \ricons{a}_j\right) + \mathit{RPA}^\mathfrak{I}\left(\ricons{s}_i, \ricons{a}_j\right).$$
The expression $\left\|\mathfrak{I}\right\|$ counts the number of attribute values required by each rule, which is a common way to measure an ABAC policy's complexity~\cite{xu2015mining, cotrini2018rhapsody}. If we instead use the objective function $\lambda \left\|\mathfrak{I}\right\| + \lossfunc(\Auth, \mathfrak{I}; \varphi^{\mathit{ABAC}}_{\NumAttVals, \NumRules})$, then the objective function penalizes not only policies that differ substantially from $\mathit{Auth}$, but also policies that are too complex. Observe that $\left\|\mathfrak{X}\right\|$ is diverse. Hence, we can use the linearity of expectation and Lemma~\ref{lem:expectations_flas} to compute all expectations needed to implement Algorithm~\ref{algo:varinfer}.

\subsection{ABAC policies from logs}
\label{sub:abac_logs}

Some miners are geared towards mining policies from logs of access requests~\cite{molloy2012generative,xu2014mining,cotrini2018rhapsody}. We now present an objective function that can be used to mine ABAC policies from access logs, instead of permission assignments. We let $\varphi := \varphi^{\mathit{ABAC}}_{M, N}$ for the rest of this subsection.\looseness=-1

A log $G$ is a disjoint union of two subsets $A$ and $D$ of $U \times P$, denoting the requests that have been \emph{authorized} and \emph{denied}, respectively.

In the case of ABAC, a policy mined from a log should aim to fulfill three requirements. The policy should be \emph{succinct}, \emph{generalize well}, and be \emph{precise}~\cite{cotrini2018rhapsody}. Therefore, we define an objective function $\lossfunc'_{\mathit{ABAC}}\left(G, \mathfrak{I}\right)$ as the sum
\begin{equation}
\lossfunc'_{\mathit{ABAC}}\left(G, \mathfrak{I}\right) = \lambda_0\left\|\mathfrak{I}\right\| + \lossfunc_1\left(G, \mathfrak{I}\right) + \lossfunc_2\left(G, \mathfrak{I}\right).
\label{eq:abac_loss_logs}
\end{equation}

The term $\left\|\mathfrak{I}\right\|$ is as defined in Section~\ref{sub:abac} and aims to make the policy succinct by penalizing complex policies. The term $\lossfunc_1\left(G, \mathfrak{I}\right)$ aims to make the mined policy generalize well and is defined as
\begin{align*}
\lossfunc_1\left(G, \mathfrak{I}\right) &= \lambda_{1, 1}\sum_{(\domi{u}, \domi{p}) \in A} \left(1 - \varphi^\mathfrak{I}(\domi{u}, \domi{p})\right) +\\
&\quad \lambda_{1, 2} \sum_{(\domi{u}, \domi{p}) \in D}\varphi^\mathfrak{I}(\domi{u}, \domi{p}).
\end{align*}

Finally, the function $\lossfunc_2\left(G, \mathfrak{I}\right)$ aims to make the mined policy precise by penalizing policies that authorize too many requests that are not in the log.

\begin{equation*}
\lossfunc_2\left(G, \mathfrak{I}\right) = \lambda_2 \sum_{(\domi{u}, \domi{p}) \in U \times P \setminus G} \varphi^\mathfrak{I}(\domi{u}, \domi{p}).
\end{equation*}

One can show that $\varphi^{\mathfrak{X}}\left(\domi{u}, \domi{p}\right)$ is diverse, for any $(\domi{u}, \domi{p}) \in U \times P$. Therefore, we can compute $\meanfield{}{\mathfrak{f} \mapsto \domi{b}}{\lossfunc'_{\mathit{ABAC}}\left(G, \mathfrak{X}\right)}$ using only the linearity of expectation and Lemma~\ref{lem:expectations_flas}.


\subsection{Business-meaningful RBAC policies}
\label{sub:business_rbac}

Frank et. al.~\cite{frank2013role} developed a probabilistic policy miner for RBAC policies that incorporated business information. Aside from a permission assignment, the miner takes as input an \emph{attribute-assignment relation} $\mathit{AA} \subseteq U \times \mathit{AVal}$, where $\mathit{AVal}$ denotes all possible combination of attribute values. It is assumed that each user is assigned exactly one combination of attribute values.\looseness=-1

This miner grants similar sets of roles to users that have similar attribute values. For this, it uses the following formula $\Delta(\domi{u},\domi{u}', \mathfrak{I})$ that measures the disagreement between the roles that a policy $\mathfrak{I}$ assigns to two users $\domi{u}$ and $\domi{u}'$:
\begin{align*}
&\Delta(\domi{u},\domi{u}', \mathfrak{I}) = \nonumber \\
&\quad \sum_{i \leq \NumRoles}\mathit{UA}^\mathfrak{I}(\domi{u}, \ricons{r}_i)\left(1 - 2\mathit{UA}^\mathfrak{I}(\domi{u}, \ricons{r}_i)\mathit{UA}^\mathfrak{I}(\domi{u}', \ricons{r}_i)\right).
\end{align*}

The formula $\left\|\mathfrak{I}\right\|$ below shows how Frank et al.'s miner measures an RBAC policy's complexity. The complexity increases whenever two users with the same combination of attribute values get assigned significantly different sets of roles.
\small
\begin{equation*}
\left\|\mathfrak{I}\right\| = \frac{1}{N} \sum_{\domi{u}, \domi{u}' \in U} \; \sum_{\ricons{a} \in \mathit{AVal}}\underline{\mathit{AA}}(\domi{u}, \ricons{a})\underline{\mathit{AA}}(\domi{u}', \ricons{a})\Delta(\domi{u},\domi{u}',\mathfrak{I}).
\label{eq:frank_total_disag}
\end{equation*}
\normalsize
Here, $N$ denotes the total of users. Note that $\underline{\mathit{AA}}$ is a rigid relation symbol representing $\mathit{AA}$. Its interpretation is therefore fixed and not computed by the policy miner. 

To mine business-meaningful RBAC policies, we use the objective function $
\lambda \left\|\mathfrak{I}\right\| + \lossfunc\left(\Auth, \mathfrak{I}; \varphi^{\mathit{RBAC}}_{N}\right)
$, where $\lambda > 0$ is a trade-off hyper-parameter. Observe that this objective function penalizes the following types of policies.
\begin{itemize}
\item Policies that assign significantly different sets of roles to users with the same attribute values.
\item Policies whose assignment of permissions to users substantially differs from the assignment given by $\mathit{Auth}$.
\end{itemize}

The random variable $\left\|\mathfrak{X}\right\|$ is, however, not diverse. This is because, for $i \leq \NumRoles$, the random fact $\mathit{UA}^\mathfrak{X}(\domi{u}, \ricons{r}_i)$ occurs more than once in $\Delta(\domi{u},\domi{u}', \mathfrak{X})$. Nonetheless, observe that
\begin{align*}
&\Delta(\domi{u},\domi{u}', \mathfrak{X}) = \nonumber \\
&\quad \sum_{i \leq \NumRoles}\mathit{UA}^\mathfrak{X}(\domi{u}, \ricons{r}_i) - 2\left(\mathit{UA}^\mathfrak{X}(\domi{u}, \ricons{r}_i)\right)^2\mathit{UA}^\mathfrak{X}(\domi{u}', \ricons{r}_i).
\end{align*}
One can then compute $\meanfield{}{\mathfrak{f} \to \domi{b}}{\Delta(\domi{u},\domi{u}', \mathfrak{X})}$ by using the linearity of expectation and the fact that $\mathbb{E}\left[X^n\right] = \left(\mathbb{E}\left[X\right]\right)^n$, for $n \in \mathbb{N}$ and $X$ a Bernoulli random variable. Hence, 
\begin{align*}
&\meanfield{}{\mathfrak{f} \to \domi{b}}{\Delta(\domi{u},\domi{u}', \mathfrak{X})} = \nonumber \\
&\quad \sum_{i \leq \NumRoles}\left(
\begin{array}{l}
\meanfield{}{\mathfrak{f} \to \domi{b}}{\mathit{UA}^\mathfrak{X}(\domi{u}, \ricons{r}_i)} - \\
2\left(\meanfield{}{\mathfrak{f} \to \domi{b}}{\mathit{UA}^\mathfrak{X}(\domi{u}, \ricons{r}_i)}\right)^2\meanfield{}{\mathfrak{f} \to \domi{b}}{\mathit{UA}^\mathfrak{X}(\domi{u}', \ricons{r}_i)}
\end{array}
\right).
\end{align*}
One can check that this observation and Lemma~\ref{lem:expectations_flas} suffice to compute the expectations necessary for Algorithm~\ref{algo:varinfer}.

\section{Mining XACML policies}
\label{sub:xacml}

Although XACML is the de facto standard for access control specification, no algorithm has previously been proposed for mining XACML policies. We now illustrate how, using {\polymver}, we have built the first XACML policy miner.


\subsection{Background}

\newcommand{\authDec}{\mathit{allow}}
\newcommand{\denyDec}{\mathit{deny}}
\newcommand{\allowov}{\mathit{AllowOv}}
\newcommand{\denyov}{\mathit{DenyOv}}
\newcommand{\firstapp}{\mathit{FirstApp}}
\newcommand{\XAttVals}{\mathit{AVals}}
\newcommand{\RequestSort}{\sortFont{REQS}}
\newcommand{\PolicySort}{\sortFont{POLS}}
\newcommand{\RuleSort}{\sortFont{RULES}}
\newcommand{\AttValSort}{\sortFont{AVALS}}

{\textbf{XACML syntax.}} To simplify the presentation, we use a reduced version of XACML, given as a BNF grammar below. However, our approach extends to the core XACML. Moreover, our reduced XACML is still powerful enough to express Continue~\cite{krishnamurthi2003continue, fisler2005verification}, a benchmark XACML policy used for policy analysis.
\begin{equation*}
\begin{array}{rcl}
\langle\texttt{Dec}\rangle &::= &\authDec \mid \denyDec\\
\langle\texttt{Rule}\rangle &::= &\left(\langle\texttt{Dec}\rangle, \alpha \right)\\
\langle\texttt{Comb}\rangle &::= &\firstapp \mid \allowov \mid \denyov\\
\langle\texttt{Pol}\rangle &::= &\left(\langle\texttt{Comb}\rangle, \left(\langle\texttt{Pol}\rangle^* \mid \langle\texttt{Rule}\rangle^*\right)\right)
\end{array}
\label{eq:xacml_bnf}
\end{equation*}

Fix a set $\XAttVals$ of attribute values. An \emph{XACML rule} is a pair $(\delta, \alpha)$, where $\delta \in \{\authDec, \denyDec\}$ is the \emph{rule's decision} and $\alpha$ is a subset of $\XAttVals$. An XACML \emph{policy} is a pair $(\kappa, \bar{\pi})$, where $\kappa \in \{\firstapp, \allowov, \denyov\}$ is a \emph{combination algorithm} and $\bar{\pi}$ is either a list of policies or a list of XACML rules. $\firstapp$, $\allowov$, $\denyov$ denote XACML's standard policy combination algorithms. We explain later how they work. For a policy $\pi$, we denote its combination algorithm by $\mathit{Comb}\left(\pi\right)$ and, for a rule $r$, we denote its decision by $\mathit{Dec}(r)$.\looseness=-1

{\textbf{XACML semantics.}} We now recall XACML's semantics. A \emph{request} is a subset of $\XAttVals$ denoting the attribute values that a subject $\domi{s}$, an action $\domi{a}$, and an object $\domi{o}$ satisfy when $\domi{s}$ attempts to execute $\domi{a}$ on $\domi{o}$. We denote by $2^\XAttVals$ the set of requests. A request \emph{satisfies} a rule $(\delta, \alpha)$ if the request contains all attributes in $\alpha$. In this case, if $\delta = \authDec$, then we say that the rule \emph{authorizes the request}; otherwise, we say that the rule \emph{denies the request}.

A policy $\pi$ of the form $(\allowov, \left(\pi'_1, \ldots, \pi'_\ell\right))$ authorizes a request $\domi{z}$ if there is an $i \leq \ell$ such that $\pi'_i$ authorizes $z$. The policy $\pi$ \emph{denies} $\domi{z}$ if no $\pi'_i$, for $i \leq \ell$, authorizes $\domi{z}$, but some $\pi'_j$, for $j \leq \ell$, denies it.

A policy $\pi$ of the form $(\denyov, \left(\pi'_1, \ldots, \pi'_\ell\right))$ denies a request if some $\pi'_i$ denies it. The policy authorizes the request if no $\pi'_i$ denies it, but some $\pi'_i$ authorizes it.

A policy $\pi$ of the form $(\firstapp, \left(\pi'_1, \ldots, \pi'_\ell\right))$ authorizes a request if there is an $i \leq \ell$ such that $\pi'_i$ authorizes it and $\pi'_j$, for $j < i$, neither authorizes it nor denies it. The policy denies a request if there is $i \leq \ell$ such that $\pi'_i$ denies it and $\pi'_j$, for $j < i$, neither authorizes it nor denies it.

\subsection{Auxiliary definitions}

For a policy $\pi =\left(\kappa, \left(\pi'_1, \ldots, \pi'_\ell\right)\right)$, we call $\pi'_i$ a \emph{child} of $\pi$. A policy is a \emph{descendant} of $\pi$ if it is a child of $\pi$ or is a descendant of a child of $\pi$.\looseness=-1

A policy $\pi$ has \emph{breadth} $N \in \mathbb{N}$ if $\ell \leq N$ and each of $\pi$'s children is either a rule or has breadth $N$. A policy $\pi$ has \emph{depth} is $M \in \mathbb{N}$ (i) if $M = 1$ and each of its children is a rule, or (ii) if $M > 1$ and some child of $\pi$ has depth $M - 1$ and the rest have depth at most $M - 1$.\looseness=-1

Two formulas $\psi_1, \psi_2 \in \Lang$ are \emph{mutually exclusive} if there is no $\mathfrak{I}$ and no $\domi{z} \in 2^{\XAttVals}$ such that both $\psi_1^\mathfrak{I}\left(\domi{z}\right)$ and $\psi_2^\mathfrak{I}\left(\domi{z}\right)$ hold. When $\psi_1$ and $\psi_2$ are mutually exclusive, we write $\psi_1 \oplus \psi_2$ instead of $\psi_1 \lor \psi_2$.


\subsection{A template formula for XACML}


For $M, N \in \mathbb{N}$, we present a template formula for the language of all XACML policies of depth and breadth at most $M$ and $N$, respectively.

Let $\mathcal{S}$ be the set of all $N$-ary sequences of length at most $M$ and let $\epsilon \in \mathcal{S}$ be the empty sequence. For $j \in \{0, \ldots, N-1\}$, we denote by $\sigma \triangleright j$ the result of appending $j$ to $\sigma$ and by $j \triangleleft \sigma$ the result of prepending $j$ to $\sigma$.

Let $\RequestSort$ be a sort representing all requests, $\AttValSort$ be a sort representing all attribute values, and $\PolicySort$ a sort representing policies and rules. For each $\sigma \in \mathcal{S}$, define a rigid constant $\ricons{y}_\sigma$ symbol of sort $\PolicySort$ such that $\ricons{y}_\sigma \neq \ricons{y}_{\sigma'}$, whenever $\sigma \neq \sigma'$.

The set of rigid constants $\{\ricons{y}_\sigma \mid \sigma \in \mathcal{S}\}$ are intended to represent a tree of XACML policies and rules. The constant $\ricons{y}_\epsilon$ is the root policy. For $\sigma \in \mathcal{S}$ with length less than $M$ and $j \in \{0, \ldots, N-1\}$, the constant $\ricons{y}_{\sigma \triangleright j}$ represents one of $\ricons{y}_{\sigma}$'s children.

Let $\vari{z}$ be a variable of sort $\RequestSort$. The formula $\varphi^{\text{XACML}}_{M, N}\left(\ricons{y}_{\epsilon}, \vari{z}\right)$ below is a template formula for the XACML fragment introduced above. We explain its main parts.

\newcommand{\xacmlAllows}{\texttt{allows}}
\newcommand{\xacmlDenies}{\texttt{denies}}
\newcommand{\lxacmlAllows}{\underline{\texttt{XAllows}}}
\newcommand{\lxacmlDenies}{\underline{\texttt{XDenies}}}
\newcommand{\ruleAllowFla}{\texttt{allowsRule}} 
\newcommand{\isRule}{\mathit{XIsRule}}
\newcommand{\polAllowFla}{\texttt{allowsPol}} 
\newcommand{\polNAFla}{\texttt{NA}} 
\newcommand{\ruleDenyFla}{\texttt{deniesRule}} 
\newcommand{\polDenyFla}{\texttt{deniesPol}} 
\newcommand{\xcomb}{\textit{Comb}}
\newcommand{\lxcomb}{\textit{XComb}}
\newcommand{\xdec}{\textit{Dec}}
\newcommand{\lxdec}{\textit{XDec}}
\newcommand{\activePol}{\mathit{XActive}}
\newcommand{\XRqA}{\mathit{XRequiresAVal}}
\newcommand{\XHasAttVal}{\underline{\mathit{hasAttVal}}}

\begin{itemize}
\item We define signature symbols that represent all terminal symbols in the BNF grammar above. For example, we define two rigid constant symbols $\lxacmlAllows$ and $\lxacmlDenies$ that represent the decisions $\authDec$ and $\denyDec$. We define two flexible function symbols $\lxdec$ and $\lxcomb$. For a rigid constant $\ricons{y}_\sigma$, $\lxdec\left(\ricons{y}_\sigma\right)$ denotes the decision of the rule represented by $\ricons{y}_\sigma$. Similarly, $\lxcomb\left(\ricons{y}_\sigma\right)$ denotes the combination algorithm of the policy represented by $\ricons{y}_\sigma$.
\item The formula $\xacmlAllows\left(\ricons{y}_\sigma, \vari{z}\right)$ holds if $\ricons{y}_\sigma$ authorizes the request represented by $\vari{z}$. The formula $\xacmlDenies\left(\ricons{y}_\sigma, \vari{z}\right)$ holds if $\ricons{y}_\sigma$ denies the request represented by $\vari{z}$ and is defined analogously. Observe that $\xacmlAllows\left(\ricons{y}_\sigma, \vari{z}\right)$ and $\xacmlDenies\left(\ricons{y}_\sigma, \vari{z}\right)$ denote formulas. Hence, $\xacmlAllows$ and $\xacmlDenies$ are not symbols in the signature we use to specify $\varphi^{\text{XACML}}_{M, N}$.
\item $\activePol$ is a flexible relation symbol and $\activePol\left(\ricons{y}_\sigma\right)$ holds if $\ricons{y}_{\sigma}$ is a descendant of $\ricons{y}_\epsilon$.
\item The formula $\polNAFla\left(\ricons{y}_\sigma, \vari{z}\right)$ holds if $\ricons{y}_\sigma$ neither authorizes nor denies the request represented by $\vari{z}$. It can be expressed in $\Lang$ as follows:
\begin{equation*}
\polNAFla\left(\ricons{y}_\sigma, \vari{z}\right) := \neg \activePol\left(\ricons{y}_\sigma\right) \lor \bigwedge_{j \leq N} \polNAFla\left(\ricons{y}_{\sigma \triangleright j}, \vari{z}\right).
\end{equation*}
\item The formula $\vari{z} \vDash \ricons{y}_\sigma$ holds if all attributes required by $\ricons{y}_\sigma$ are contained by the request represented by $\vari{z}$. This formula can be expressed in $\Lang$ as follows:
\begin{equation*}
\bigwedge_{\domi{a} \in \XAttVals} \left(\XRqA\left(\ricons{y}_\sigma, \ricons{a}\right) \to \XHasAttVal\left(\vari{z}, \ricons{a}\right)\right),
\end{equation*}
where $\XRqA$ is a flexible relation symbol and $\XHasAttVal$ and $\ricons{a}$, for $\domi{a} \in \XAttVals$, are rigid symbols. For a policy $\mathfrak{I}$, $\XRqA^{\mathfrak{I}}\left(\ricons{y}_\sigma, \ricons{a}\right)$ holds if $\ricons{y}_\sigma$ is a rule and requires attribute $\domi{a}$ to be satisfied. The formula $\XHasAttVal\left(\vari{z}, \ricons{a}\right)$ checks if the request contains attribute $\domi{a}$.
\end{itemize}

\begin{align*}
\varphi^{\text{XACML}}_{M, N}\left(\ricons{y}_{\epsilon}, \vari{z}\right) &:= \xacmlAllows\left(\ricons{y}_{\epsilon}, \vari{z}\right). \\
\xacmlAllows\left(\ricons{y}_{\sigma}, \vari{z}\right) &:= \left(\isRule\left(\ricons{y}_{\sigma}\right) \to \ruleAllowFla\left(\ricons{y}_{\sigma}, \vari{z}\right)\right) \land \\
& \quad \left(\neg \isRule\left(\ricons{y}_{\sigma}\right) \to \polAllowFla\left(\ricons{y}_{\sigma}, \vari{z}\right)\right).\\
\ruleAllowFla\left(\ricons{y}_\sigma, \vari{z}\right) &:= \activePol\left(\ricons{y}_\sigma\right) \land \lxdec\left(\ricons{y}_\sigma\right) = \underline{\authDec} \,\land \vari{z} \vDash \ricons{y}_\sigma.\\
\end{align*}
\vspace{-30pt}
\begin{align*}
&\polAllowFla\left(\ricons{y}_\sigma, \vari{z}\right) := \activePol\left(\ricons{y}_\sigma\right) \land\\[5pt]
&\quad \left[\begin{array}{l}
	\left(\begin{array}{l}
	\lxcomb\left(\ricons{y}_\sigma\right) = \underline{\allowov} \, \land\\ 
	\bigvee_{j \leq N} \xacmlAllows\left(\ricons{y}_{\sigma \triangleright j}, \vari{z}\right)
\end{array}\right) \oplus\\[15pt]
\left(\begin{array}{l}
	\lxcomb\left(\ricons{y}_\sigma\right) = \underline{\firstapp} \, \land\\[5pt]
	\bigoplus_{j \leq N} \left(
	\begin{array}{l}
		\bigwedge_{i < j} \polNAFla\left(\ricons{y}_{\sigma \triangleright i}, \vari{z}\right) \land\\
		\xacmlAllows\left(\ricons{y}_{\sigma \triangleright j}, \vari{z}\right)
	\end{array}\right)
\end{array}\right) \oplus\\[25pt]
\left(\begin{array}{l}
	\lxcomb\left(\ricons{y}_\sigma\right) = \underline{\denyov} \, \land\\[5pt]
	\displaystyle\bigoplus_{j \leq N} \left( 
	\begin{array}{l}
		\bigwedge_{i < j} \polNAFla\left(\ricons{y}_{\sigma \triangleright i}, \vari{z}\right) \land\\
		\xacmlAllows\left(\ricons{y}_{\sigma \triangleright j}, \vari{z}\right) \land\\
		\bigwedge_{i < k} \neg \xacmlDenies\left(\ricons{y}_{\sigma \triangleright k}, \vari{z}\right)
	\end{array}
	\right)
\end{array}\right)
\end{array}\right].
\end{align*}
\begin{lemma}
Formula $\varphi^{\text{XACML}}_{M, N}$ is a template formula for the language of all XACML policies of depth and breadth at most $M$ and $N$, respectively.
\label{lem:xacml_task_1}
\end{lemma}

\begin{proof}
We define a mapping $\mathcal{M}$ from interpretation functions to XACML policies using an auxiliary mapping $\mathcal{M}'$. For a sequence $\sigma \in \mathcal{S}$, we inductively define $\mathcal{M}'\left(\mathfrak{I}, \sigma\right)$ as follows:
\begin{itemize}
\item If $\sigma$ has length $M$ or $\activePol^{\mathfrak{I}}\left(\ricons{y}_{\sigma \triangleright j}\right) = 0$, for all $j \leq N$, then
\begin{equation*}
\mathcal{M}'\left(\mathfrak{I}, \sigma\right) = \left(\lxdec^{\mathfrak{I}}\left(\ricons{y}_\sigma\right), \mathit{AV} \right),
\label{eq:m_i_sigma_base}
\end{equation*}
where,
\begin{equation*}
\mathit{AV} = \left\{a \in \XAttVals \mid \XRqA^{\mathfrak{I}}\left(\ricons{y}_\sigma, \ricons{a}\right)\right\}.
\label{eq:m_i_sigma_base}
\end{equation*}
\item Otherwise,
\begin{equation*}
\mathcal{M}'\left(\mathfrak{I}, \sigma\right) = \left(\lxcomb^{\mathfrak{I}}\left(\ricons{y}_\sigma\right), \left(\mathcal{M}'\left(\mathfrak{I}, \sigma \triangleright j\right)_{j \leq N}\right)\right).
\label{eq:m_i_sigma_ind}
\end{equation*}
\end{itemize}
For an interpretation function $\mathfrak{I}$, we define $\mathcal{M}\left(\mathfrak{I}\right) = \mathcal{M}'\left(\mathfrak{I}, \epsilon\right)$. We show that $\mathcal{M}$ is surjective. Let $\pi$ be a XACML policy. 
For $\sigma \in \mathcal{S}$ and $\pi'$ a descendant of $\pi$, we inductively define the following policy:
\begin{equation*}
\pi'[\sigma] := 
\begin{cases}
\pi' &\text{if $\sigma = \epsilon$ and}\\
\pi'_i[\sigma'] &\text{if $\sigma = i \triangleleft \sigma'$ and $\pi' = \left(\kappa, \left(\pi'_1, \ldots, \pi'_k\right)\right)$.}
\end{cases}
\end{equation*}
%

We now present an interpretation function $\mathfrak{I}$ such that $\mathcal{M}\left(\mathfrak{I}\right) = \pi$. Let $\sigma \in \mathcal{S}$ and $\bot$ be any arbitrary value. Then 

\begin{equation*}
\activePol^{\mathfrak{I}}\left(\ricons{y}_{\sigma}\right) 
\begin{cases}
1 & \text{if there is a descendant $\pi'$ of $\pi$}\\
& \quad \text{with $\pi\left[\sigma\right] = \pi'$ and}\\
0 & \text{otherwise.}
\end{cases}
\label{eq:active_pol_i}
\end{equation*}

\begin{equation*}
\lxcomb^{\mathfrak{I}}\left(\ricons{y}_{\sigma}\right)
\begin{cases}
\xcomb\left(\pi[\sigma]\right) & \text{if $\pi\left[\sigma\right]$ is a policy and}\\
\bot & \text{otherwise.}
\end{cases}
\label{eq:comb_i}
\end{equation*}

\begin{equation*}
\lxdec^{\mathfrak{I}}\left(\ricons{y}_{\sigma}\right)
\begin{cases}
\xdec\left(\pi[\sigma]\right) & \text{if $\pi\left[\sigma\right]$ is a rule and}\\
\bot & \text{otherwise.}
\end{cases}
\label{eq:dec_i}
\end{equation*}

\begin{equation*}
\isRule^{\mathfrak{I}}\left(\ricons{y}_{\sigma}\right)
\begin{cases}
1 & \text{if $\pi\left[\sigma\right]$ is a rule and}\\
0 & \text{otherwise.}
\end{cases}
\label{eq:is_rule_i}
\end{equation*}

\begin{equation*}
\XRqA^{\mathfrak{I}}\left(\ricons{y}_{\sigma}, \ricons{a}\right)
\begin{cases}
1 & \text{if $\pi[\sigma]$ is a rule of the form $\left(d, \alpha\right)$,}\\
  & \text{$\alpha \subseteq \XAttVals$, and $a \in \alpha$}\\
  &\\
0 & \text{otherwise.}
\end{cases}
\label{eq:requires_aval_i}
\end{equation*}

It is straightforward to verify that $\mathcal{M}\left(\mathfrak{I}\right) = \pi$.
\end{proof}

Having a template formula for this XACML fragment, we now define an objective function. 
An example of an objective function is $
\lambda \left\|\mathfrak{I}\right\| + \lossfunc\left(\Auth, \mathfrak{I}; \varphi^{\text{XACML}}_{M, N}\right)$, where $\lambda > 0$ is a hyper-parameter and $\left\|\mathfrak{I}\right\|$ defines $\mathfrak{I}$'s complexity. We inductively define the complexity $\mathit{compl}\left(\pi\right)$ of a XACML policy $\pi$ as follows.
\begin{itemize}
\item If $\pi=\left(\delta, \alpha\right)$, then $\mathit{compl}\left(\pi\right) = \left|\alpha\right|$.
\item If $\pi=\left(\kappa, \left(\pi_1, \ldots, \pi_k\right)\right)$, then $\mathit{compl}\left(\pi\right) = \left|\alpha\right|$.
\end{itemize}
Finally, we define $\left\|\mathfrak{I}\right\|$ as $\mathit{compl}\left(\mathcal{M}\left(\mathfrak{I}\right)\right)$.

\subsection{Computing expectations} 

For a formula $\varphi \in \Lang$ and a request $\domi{z} \in 2^{\XAttVals}$, we define the random variable $\varphi^{\mathfrak{X}}\left(\domi{z}\right)$ in a way similar to the one given in Section~\ref{sub:aux_defs}. We now give some auxiliary definitions that help to compute $\meanfield{}{\mathfrak{f} \to \domi{b}}{\left(\varphi^{\text{XACML}}_{M, N}\right)^{\mathfrak{X}}\left(\domi{z}\right)}$.

\begin{lemma}
Let $\domi{z} \in 2^\XAttVals$ and $\psi_1, \psi_2$ be mutually exclusive formulas, then
\begin{equation*}
\meanfield{V \neq \mathfrak{f}}{\mathfrak{f} \mapsto \domi{b}} {\left(\psi_1 \oplus \psi_2\right)^\mathfrak{X}\left(\domi{z}\right)} = 
\meanfield{V \neq \mathfrak{f}}{\mathfrak{f} \mapsto \domi{b}} {\psi_1^\mathfrak{X}\left(\domi{z}\right)} + \meanfield{V \neq \mathfrak{f}}{\mathfrak{f} \mapsto \domi{b}} {\psi_2^\mathfrak{X}\left(\domi{z}\right)}.
\label{eq:exp_mutually_disjoint}
\end{equation*}
\label{lem:disjoint_flas}
\end{lemma}

\begin{proof}
Note that $\left(\psi_1 \oplus \psi_2\right)^{\mathfrak{I}}\left(\domi{z}\right) = 1$ iff either $\psi_1^{\mathfrak{I}}\left(\domi{z}\right) = 1$ or $\psi_2^{\mathfrak{I}}\left(\domi{z}\right) = 1$.
\begin{align*}
\meanfield{V \neq \mathfrak{f}}{\mathfrak{f} \mapsto \domi{b}} {\left(\psi_1 \oplus \psi_2\right)^\mathfrak{X}\left(\domi{z}\right)}
&= \sum_{\mathfrak{I}}q\left(\mathfrak{I}\right) \left(\psi_1 \oplus \psi_2\right)^\mathfrak{I}\left(\domi{z}\right)\\
&= \sum_{\mathfrak{I}: \left(\psi_1 \oplus \psi_2\right)^\mathfrak{I}\left(\domi{z}\right) = 1}q\left(\mathfrak{I}\right)\\
&= \sum_{\mathfrak{I}: \psi_1^\mathfrak{I}\left(\domi{z}\right) = 1}q\left(\mathfrak{I}\right) + \sum_{\mathfrak{I}: \psi_2^\mathfrak{I}\left(\domi{z}\right) = 1}q\left(\mathfrak{I}\right)\\
&= \sum_{\mathfrak{I}}q\left(\mathfrak{I}\right) \psi_1^\mathfrak{I}\left(\domi{z}\right) + \sum_{\mathfrak{I}}q\left(\mathfrak{I}\right) \psi_2^\mathfrak{I}\left(\domi{z}\right)\\
&= \meanfield{V \neq \mathfrak{f}}{\mathfrak{f} \mapsto \domi{b}} {\psi_1^\mathfrak{X}\left(\domi{z}\right)} + \meanfield{V \neq \mathfrak{f}}{\mathfrak{f} \mapsto \domi{b}} {\psi_2^\mathfrak{X}\left(\domi{z}\right)}.
\end{align*}
\end{proof}

\begin{definition}
A set $\Phi \subseteq \Lang$ of formulas is \emph{unrelated} if for every $\varphi \in \Phi$ and every atomic formula $\alpha$ occurring in $\varphi$, there is no $\psi \in \Phi \setminus \left\{\varphi\right\}$ such that $\alpha$ occurs in $\varphi$.
\label{def:unrelated_flas}
\end{definition}

\begin{lemma}
If $\domi{z} \in 2^\XAttVals$ and $\Phi$ is a set of unrelated formulas, then $\left\{\varphi^{\mathfrak{X}}\left(\domi{z}\right) \mid \varphi \in \Phi\right\}$, under the distribution $q$, is mutually independent.
\label{lem:unrelated_independent}
\end{lemma}

\begin{proof}
For simplicity, we assume that $\Phi = \left\{\varphi_1, \varphi_2\right\}$. The proof for the general case is analogous.

Observe that, since $\Phi$ is unrelated, an interpretation function $\mathfrak{I}$ can be regarded as the union of two interpretation functions $\mathfrak{I}_1$ and $\mathfrak{I}_2$ where $\mathfrak{I}_1$ interprets the atomic formulas occurring in $\varphi_1$ and $\mathfrak{I}_2$ interprets those in $\varphi_2$. Consequently, the distribution $q\left(\mathfrak{I}\right)$ can be factorized as $q_1\left(\mathfrak{I}_1\right)q_2\left(\mathfrak{I}_2\right)$, where $q_i$, for $i \leq 2$, is the marginal mean-field approximating the joint distribution of the random facts of $\varphi_i$.

For an event $A$, let $\prob_q\left(A\right)$ denote the probability of $A$ under the distribution $q$. To prove the lemma, it suffices to show, for $\domi{b}_1, \domi{b}_2 \in \{0, 1\}$, that 
\begin{equation}
\begin{aligned}
&\prob_q\left(\varphi_1^{\mathfrak{X}}\left(\domi{z}\right) = \domi{b}_1, \varphi_2^{\mathfrak{X}}\left(\domi{z}\right) = \domi{b}_2\right) =\\
&\qquad \qquad \prob_q\left(\varphi_1^{\mathfrak{X}}\left(\domi{z}\right) = \domi{b}_1\right)\prob_q\left(\varphi_2^{\mathfrak{X}}\left(\domi{z}\right) = \domi{b}_2\right),
\end{aligned}
\end{equation}
which implies that $\Phi$ is mutually independent.

\begin{align*}
&\prob_q\left(\varphi_1^{\mathfrak{X}}\left(\domi{z}\right) = \domi{b}_1, \varphi_2^{\mathfrak{X}}\left(\domi{z}\right) = \domi{b}_2\right) \\
&= \sum_{\substack{\mathfrak{I} : \varphi_1^{\mathfrak{I}}\left(\domi{z}\right) = \domi{b}_1\\ \quad \varphi_2^{\mathfrak{I}}\left(\domi{z}\right) = \domi{b}_2}} q\left(\mathfrak{I}\right)\\
&= \sum_{\substack{\mathfrak{I} : \varphi_1^{\mathfrak{I}}\left(\domi{z}\right) = \domi{b}_1\\ \quad \varphi_2^{\mathfrak{I}}\left(\domi{z}\right) = \domi{b}_2}} q_1\left(\mathfrak{I}_1\right)q_2\left(\mathfrak{I}_2\right)\\
&= \sum_{\mathfrak{I}_1 : \varphi_1^{\mathfrak{I}_1}\left(z\right) = \domi{b}_1}\sum_{\mathfrak{I}_2 : \varphi_2^{\mathfrak{I}_2}\left(\domi{z}\right) = \domi{b}_2} q_1\left(\mathfrak{I}_1\right)q_2\left(\mathfrak{I}_2\right)\\
&= \left(\sum_{\mathfrak{I}_1 : \varphi_1^{\mathfrak{I}_1}\left(\domi{z}\right) = \domi{b}_1}q_1\left(\mathfrak{I}_1\right)\right)\left(\sum_{\mathfrak{I}_2 : \varphi_2^{\mathfrak{I}_2}\left(\domi{z}\right) = \domi{b}_2} q_2\left(\mathfrak{I}_2\right)\right)\\
&= \left(\sum_{\substack{\mathfrak{I}_1,\mathfrak{I}_2\\\mathfrak{I}_1 : \varphi_1^{\mathfrak{I}_1}\left(\domi{z}\right) = \domi{b}_1}}q_1\left(\mathfrak{I}_1\right)q_2\left(\mathfrak{I}_2\right)\right)\left(\sum_{\substack{\mathfrak{I}_1, \mathfrak{I}_2\\\mathfrak{I}_2 : \varphi_2^{\mathfrak{I}_2}\left(z\right) = \domi{b}_2}} q_1\left(\mathfrak{I}_1\right)q_2\left(\mathfrak{I}_2\right)\right)\\
&= \left(\sum_{\mathfrak{I} : \varphi_1^{\mathfrak{I}}\left(\domi{z}\right) = \domi{b}_1} q\left(\mathfrak{I}\right)\right)
\left(\sum_{\mathfrak{I} : \varphi_2^{\mathfrak{I}}\left(\domi{z}\right) = \domi{b}_2} q\left(\mathfrak{I}\right)\right)\\
&= \prob_q\left(\varphi_1^{\mathfrak{X}}\left(\domi{z}\right) = \domi{b}_1\right)\prob_q\left(\varphi_2^{\mathfrak{X}}\left(\domi{z}\right) = \domi{b}_2\right).
\end{align*}
\end{proof}

\begin{lemma}
We can compute $\meanfield{}{\mathfrak{f} \to \domi{b}}{\left(\varphi^{\text{XACML}}_{M, N}\right)^{\mathfrak{X}}\left(\domi{z}\right)}$ using only the equations given in Lemmas~\ref{lem:pre_expectations_flas} and~\ref{lem:disjoint_flas}.
\end{lemma}

\begin{proof}
Observe that every atomic formula in $\ruleAllowFla\left(\ricons{y}_\sigma, \vari{z}\right)$ occurs exactly once, so $\ruleAllowFla^{\mathfrak{X}}\left(\ricons{y}_\sigma, \domi{z}\right)$ is diverse. Hence, by Corollary~\ref{cor:occurs_once}, we can use Lemma~\ref{lem:expectations_flas} to compute the expectation $\meanfield{}{\mathfrak{f}\to \domi{b}}{\ruleAllowFla^{\mathfrak{X}}\left(\ricons{y}_\sigma, \domi{z}\right)}$.


The formula $\polAllowFla\left(\ricons{y}_\sigma, \vari{z}\right)$ can be rewritten as
\begin{equation*}
\begin{aligned}
\left(\activePol\left(\ricons{y}_\sigma\right) \land \psi_1\left(\ricons{y}_\sigma, \vari{z}\right)\right) \oplus
\left(\activePol\left(\ricons{y}_\sigma\right) \land \psi_2\left(\ricons{y}_\sigma, \vari{z}\right)\right) \oplus\\\
\left(\activePol\left(\ricons{y}_\sigma\right) \land \psi_3\left(\ricons{y}_\sigma, \vari{z}\right)\right).
\end{aligned}
\end{equation*}
Each formula $\psi_i\left(\ricons{y}_\sigma, \vari{z}\right)$ is built from a set of unrelated formulas. Hence, by Lemma~\ref{lem:unrelated_independent}, we can use Lemma~\ref{lem:pre_expectations_flas} to compute $\meanfield{}{\mathfrak{f}\to \domi{b}}{\psi_i\left(\ricons{y}_\sigma, \vari{z}\right)}$. 
Using Lemmas~\ref{lem:disjoint_flas} and~\ref{lem:pre_expectations_flas}, we can show that
\begin{align*}
&\meanfield{}{\mathfrak{f}\to \domi{b}}{\polAllowFla^{\mathfrak{X}}\left(\ricons{y}_\sigma, \domi{z}\right)} = \\
&\qquad \meanfield{}{\mathfrak{f}\to \domi{b}}{\activePol^{\mathfrak{X}}\left(\ricons{y}_\sigma, \domi{z}\right)} \times \left(
\begin{array}{l}
\meanfield{}{\mathfrak{f}\to \domi{b}}{\psi_1^{\mathfrak{X}}\left(\ricons{y}_\sigma, \domi{z}\right)} +\\ 
\meanfield{}{\mathfrak{f}\to \domi{b}}{\psi_2^{\mathfrak{X}}\left(\ricons{y}_\sigma, \domi{z}\right)} +\\ 
\meanfield{}{\mathfrak{f}\to \domi{b}}{\psi_3^{\mathfrak{X}}\left(\ricons{y}_\sigma, \domi{z}\right)}
\end{array}
\right).
\end{align*}
Therefore, $\meanfield{}{\mathfrak{f}\to \domi{b}}{\polAllowFla^{\mathfrak{X}}\left(\ricons{y}_\sigma, \domi{z}\right)}$ can be computed using only Lemmas~\ref{lem:disjoint_flas} and~\ref{lem:pre_expectations_flas}.

Finally, recall that $\varphi^{\text{XACML}}_{M, N}\left(\ricons{y}_{\epsilon}, \vari{z}\right) = \xacmlAllows\left(\ricons{y}_{\epsilon}, \vari{z}\right)$. Observe now that $\xacmlAllows\left(\ricons{y}_{\epsilon}, \vari{z}\right)$ is built from the following unrelated set: 
\begin{equation*}
\left\{\isRule\left(\ricons{y}_{\sigma}\right), \ruleAllowFla\left(\ricons{y}_{\sigma}, \vari{z}\right), \polAllowFla\left(\ricons{y}_{\sigma}, \vari{z}\right)\right\}.
\end{equation*}
By Lemma~\ref{lem:unrelated_independent}, the corresponding set of random variables is independent. Hence, we can use Lemma~\ref{lem:pre_expectations_flas} to reduce the computation of $\meanfield{}{\mathfrak{f} \to \domi{b}}{\left(\varphi^{\text{XACML}}_{M, N}\right)^{\mathfrak{X}}\left(\domi{z}\right)}$ to the computation of $\meanfield{}{\mathfrak{f} \to \domi{b}}{\isRule^{\mathfrak{X}}\left(\ricons{y}_{\sigma}\right)}$, $\meanfield{}{\mathfrak{f} \to \domi{b}}{\ruleAllowFla^{\mathfrak{X}}\left(\ricons{y}_{\sigma}, \domi{z}\right)}$, and $\meanfield{}{\mathfrak{f} \to \domi{b}}{\polAllowFla^{\mathfrak{X}}\left(\ricons{y}_{\sigma}, \domi{z}\right)}$. However, as observed above, all these expectations can be computed using Lemmas~\ref{lem:pre_expectations_flas} and~\ref{lem:disjoint_flas}. Hence, we can compute $\meanfield{}{\mathfrak{f} \to \domi{b}}{\left(\varphi^{\text{XACML}}_{M, N}\right)^{\mathfrak{X}}\left(\domi{z}\right)}$ using only those two lemmas.
\end{proof}

Having proven the previous lemmas, we can now implement Algorithm~\ref{algo:varinfer} to produce a policy miner for XACML policies.

\section{Modeling RBAC temporal constraints}
\label{sec:app_uni_modeling_temp_constraints}

\subsection{Periodic expressions}
\label{sub:periodic_expressions}

We recall here the definition of temporal constraints and then explain how we model them in our language $\mathcal{L}$. We start by defining a periodic expression. In order to prove some results, we use a formal definition instead of the original definition~\cite{bertino2001trbac}.

\begin{definition}
A \emph{calendar sequence} is a tuple $C = \left(C_0, \ldots, C_n\right)$ of strings. Each $C_i$, for $i \leq n$, is called a \emph{calendar}.
\end{definition}

Each string in a calendar sequence denotes a time unit. A standard example of a calendar sequence is $C_s =$ ($``\domi{Months}''$, $``\domi{Days}''$, $``\domi{Hours}''$, $``\domi{Hours}''$). The string $``\domi{Hours}''$ intentionally occurs twice.

\begin{definition}
Let $C = \left(C_i\right)_{i \leq n}$ be a calendar sequence. A \emph{$C$-periodic expression} is a tuple 
$$\left(O_1, O_2, \ldots, O_{n-1}, \domi{w}\right) \in \left(2^{\mathbb{N}}\right)^{n-1} \times \mathbb{N}.$$ 
The set of all $C$-periodic expressions is denoted by $\mathit{PEx}\left(C\right)$.
\end{definition}

Intuitively, for $i < n$, the set $O_i$ represents a set of time units in $C_i$ and $r$ is a time length, measured with the unit $C_n$. We give an example.
%

\begin{example}
Let $C_s = \left(``\domi{Months}'', ``\domi{Days}'', ``\domi{Hours}'', ``\domi{Hours}''\right)$. The first eight hours of the first and fifth day of each even month can be represented with the following $C_s$-periodic expression:
$$\left(\left\{2, 4, \ldots, 12\right\}, \left\{1, 5\right\}, \left\{1\right\}, 8\right).$$
Here, the first set in the tuple denotes the even months, the second set denotes the first and fifth day, the third set denotes the first hour, and the last number denotes the 8-hour length.
\label{ex:uni_first_ten_hours}
\end{example}

\begin{definition}
Let $C = \left(C_i\right)_{i \leq n}$ be a calendar sequence. A \emph{$C$-time instant} $\domi{t}$ is a tuple in $\mathbb{N}^{n-1}$. A $C$-time instant $\domi{t} = \left(\domi{t}_i\right)_{i < n}$ \emph{satisfies} a $C$-periodic expression $\left(O_1, \ldots, O_{n-1}, \domi{w}\right)$ if all of the following hold:
\begin{itemize}
\item For $i < n - 1$, $\domi{t}_i \in O_i$.
\item There is $\domi{t}'_{n - 1} \in O_{n-1}$ such that $\domi{t}'_{n-1} \leq \domi{t}_{n-1} < \domi{t}'_{n-1} \oplus \domi{w}$. Here, $\domi{t}'_{n-1} \oplus \domi{w}$ is the result of transforming $\domi{t}'_{n-1}$ and $\domi{w}$ to a common time unit and then adding them.
\end{itemize}
\label{def:uni_time_instant_satisfies}
\end{definition}

\begin{example}
Let $P$ be the $C_s$-periodic expression from Example~\ref{ex:uni_first_ten_hours}. Then the $C_s$-time instant $(4, 1, 2)$ (i.e., the second hour of the first day of April) satisfies $P$, whereas the $C_s$-time instant $(5, 1, 2)$ (i.e., the second hour of the first day of May) does not.
\label{ex:uni_first_ten_hours_satisfies}
\end{example}

\subsection{Formalizing periodic expressions in $\Lang$}
\label{sub:formalizing_periodic_expr_lang}

Let $C$ be a calendar sequence. We now show how to formalize periodic expressions in $\mathit{PEx}\left(C\right)$ in our language $\mathcal{L}$. For illustration, we show how to do this when $C = C_s$, but the general case is analogous. We start by defining a signature for defining temporal constraints. We first formally introduce the components of this signature and then give some intuition.

\begin{definition}
Let $\Sigma_{t}$ be a signature containing the following:
\begin{itemize}
\item A sort $\sortFont{INSTANTS}$ for denoting $C_s$-time instants.
\item Three sorts $\sortFont{MONTHS}, \sortFont{DAYS}, \sortFont{HOURS}$ for denoting months, days, and hours, respectively.
\item Three flexible unary relation symbols:
\begin{itemize}
\item $\mathit{PM} : \sortFont{MONTHS}$.
\item $\mathit{PD} : \sortFont{DAYS}$.
\item $\mathit{PH} : \sortFont{HOURS}$.
\end{itemize}
\item Three rigid unary function symbols:
\begin{itemize}
\item $\underline{\mathit{monthOf}} : \sortFont{INSTANTS} \to \sortFont{MONTHS}$.
\item $\underline{\mathit{dayOf}} : \sortFont{INSTANTS} \to \sortFont{DAYS}$.
\item $\underline{\mathit{hourOf}} : \sortFont{INSTANTS} \to \sortFont{HOURS}$.
\end{itemize}
\item Rigid binary relation symbols $\leq : \sortFont{HOURS} \times \sortFont{HOURS}$ and $< : \sortFont{HOURS} \times \sortFont{HOURS}$.
\item Rigid constant symbols $\ricons{1}, \ricons{2}, \ldots, \ricons{12}$ of sort $\sortFont{MONTHS}$.
\item Rigid constant symbols $\ricons{1}, \ricons{2}, \ldots, \ricons{31}$ of sort $\sortFont{DAYS}$.
\item Rigid constant symbols $\ricons{1}, \ricons{2}, \ldots, \ricons{24}$ of sort $\sortFont{HOURS}$.
\end{itemize}
\label{def:signat_temp_constraints}
\end{definition}

For an interpretation function $\mathfrak{I}$, the sets $\mathit{PM}^\mathfrak{I}$, $\mathit{PD}^\mathfrak{I}$, and $\mathit{PH}^{\mathfrak{I}}$ define the first three entries of a $C_s$-periodic expression. The functions $\underline{\mathit{monthOf}}^{\mathfrak{I}}$, $\underline{\mathit{dayOf}}^{\mathfrak{I}}$, and $\underline{\mathit{hourOf}}^{\mathfrak{I}}$ compute, respectively, the month, day, and hour of a time instant. Recall that $\underline{\mathit{monthOf}}$, $\underline{\mathit{dayOf}}$, and $\underline{\mathit{hourOf}}$ are rigid symbols, so $\underline{\mathit{monthOf}}^{\mathfrak{I}}$, $\underline{\mathit{dayOf}}^{\mathfrak{I}}$, and $\underline{\mathit{hourOf}}^{\mathfrak{I}}$ do not depend on $\mathfrak{I}$.

Let $\sortFont{INSTANTS}$ be a sort denoting $C_s$-time instants, $\vari{t}$ be a variable of sort $\sortFont{INSTANTS}$, and $\varphi\left(\vari{t}\right)$ be the following formula:
%
%
\begin{align*}
\left(\bigvee_{1 \leq \domi{m} \leq 12} \left(\mathit{PM}\left(\ricons{m}\right) \land \underline{\mathit{monthOf}}\left(\vari{t}\right) = \ricons{m}\right)\right) \land \\
\left(\bigvee_{1 \leq \domi{d} \leq 31} \left(\mathit{PD}\left(\ricons{d}\right) \land \underline{\mathit{dayOf}}\left(\vari{t}\right) = \ricons{d}\right)\right) \land \\
\left(\bigvee_{1 \leq \domi{h} \leq 24} \left(\mathit{PH}\left(\ricons{h}\right) \land \ricons{h} \leq \underline{\mathit{hourOf}}\left(\vari{t}\right) < \ricons{h} + \cons{w}\right)\right).
\end{align*}

\begin{definition}
We define $\mathcal{M}$ as the following mapping from interpretation functions to $\mathit{PEx}\left(\mathcal{C}_s\right)$ defined as follows. For an interpretation function $\mathfrak{I}$, $\mathcal{M}\left(\mathfrak{I}\right) := \left(M_{\mathfrak{I}}, D_{\mathfrak{I}}, H_{\mathfrak{I}}, 1\right)$, where
\begin{itemize}
\item $M_{\mathfrak{I}} := \left\{\domi{m} \mid \domi{m} \in \mathit{PM}^{\mathfrak{I}}\right\}$,
\item $D_{\mathfrak{I}} := \left\{\domi{d} \mid \domi{d} \in \mathit{PD}^{\mathfrak{I}}\right\}$,
\item $H_{\mathfrak{I}} := \left\{\domi{h} \mid \domi{h}' \leq \domi{h} < \domi{h}' + \cons{w}^{\mathfrak{I}}, \text{ for some $\domi{h}' \in \mathit{PH}^{\mathfrak{I}}$}\right\}$.
\end{itemize}
\label{def:uni_interp_mapping}
\end{definition}

We use $\mathcal{M}$ to prove Theorem~\ref{thm:uni_temp_constraints_translation} below, which claims that $\varphi$ is a template formula for the set of periodic expressions. Note that $\mathcal{M}$ is not surjective on $\mathit{PEx}\left(\mathcal{C}_s\right)$. However, any periodic expression $(O_1, O_2, O_3, \domi{w})$ is equivalent to an expression of the form $(O_1, O_2, O^{\domi{w}}_3, 1)$, where $O^{\domi{w}}_3 := \left\{\domi{o} + \domi{w'} \mid \domi{o} \in \domi{O}_3, \domi{w}' < \domi{w}\right\}.$ Therefore, although $\mathcal{M}$ is not surjective, it is expressive enough to capture all periodic expressions up to equivalence.

\begin{theorem}
For every interpretation function $\mathfrak{I}$, a $C_s$-time instant $\domi{t} = \left(\domi{m}, \domi{d}, \domi{h}\right)$ satisfies $\mathcal{M}\left(\mathfrak{I}\right) = \left(M_{\mathfrak{I}}, D_{\mathfrak{I}}, H_{\mathfrak{I}}, 1\right)$ iff $\domi{t} \in \varphi^{\mathfrak{I}}$.
\label{thm:uni_temp_constraints_translation}
\end{theorem}

\begin{proof}
\begin{align*}
&\text{$\domi{t}$ satisfies $\mathcal{M}\left(\mathfrak{I}\right)$}\\
&\Leftrightarrow \domi{m} \in M_{\mathfrak{I}}, \domi{d} \in D_{\mathfrak{I}}, \domi{h} \in H_{\mathfrak{I}}\\
&\Leftrightarrow \domi{m} \in \mathit{PM}^{\mathfrak{I}}, \domi{d} \in \mathit{PD}^{\mathfrak{I}}, \text{ there is $\domi{h}' \in \mathit{PH}^{\mathfrak{I}}$ s.t. $\domi{h}' \leq \domi{h} < \domi{h}' + \domi{w}^{\mathfrak{I}}$}\\
&\Leftrightarrow 
\left(
\begin{array}{l}
\domi{m} \in \mathit{PM}^{\mathfrak{I}}, \, \underline{\mathit{monthOf}}^{\mathfrak{I}}\left(\domi{t}\right) = \domi{m},\\
\domi{d} \in \mathit{PD}^{\mathfrak{I}}, \, \underline{\mathit{dayOf}}^{\mathfrak{I}}\left(\domi{t}\right) = \domi{d},\\
\text{ there is $\domi{h}' \in \mathit{PH}^{\mathfrak{I}}$ s.t. $\domi{h}' \leq \domi{h} < \domi{h}' + \domi{w}^{\mathfrak{I}}$}, \, \underline{\mathit{hourOf}}^{\mathfrak{I}}\left(\domi{t}\right) = \domi{h}\\
\end{array}
\right)\\
&\Leftrightarrow
\domi{t} \in \varphi^{\mathfrak{I}}.
\end{align*}
\end{proof}

\section{Datasets and synthetic policies used for experiments}

\subsection{Datasets for BM-RBAC}\label{sub:bm_rbac_datasets} We use the access control matrix provided by Amazon for the IEEE MLSP 2012 competition~\cite{mlsp2012amazon}. They assign three types of permissions, named ``HOST'', ``PERM\_GROUP'', and ``SYSTEM\_GROUP''. For each type of permission, we created an access control matrix by collecting all users and all permissions belonging to that type. There are approximately 30,000 users, 1,700 permissions of type ``HOST'', 6,000 of type ``PERM\_GROUP'', and 20,000 of type ``SYSTEM\_GROUP''.
 
The resulting access control matrices are far too large to be handled efficiently by the policy miner we developed. To address this, during 5-fold cross-validation (see Section~\ref{sub:evaluation} for an overview), we worked instead with an access control submatrix induced by a sample of 30\% of all users. Each fold used a different sample of users. To see why this is enough, we remark that, in RBAC policies, the number $R$ of roles is usually much smaller than the number $N$ of users. Moreover, the number $K$ of possible subsets of permissions that users are assigned by RBAC policies is small in comparison to the whole set of possible subsets of permissions. If $N$ is much larger than $K$, then, by the pidgeonhole principle, many users have the same subset of permissions. Therefore, it is not necessary to use all $N$ users to mine an adequate RBAC policy, as only a fraction of them has all the necessary information. The high TPR (above 80\%) of the policy that we mined supports the fact that using a submatrix is still enough to mine policies that generalize well.

\subsection{Synthetic policy for spatio-temporal RBAC}
\label{sub:synth_pol_starbac}
We present here the synthetic spatio-temporal RBAC policy that we used for our experiments. We assume the existence of five rectangular buildings, described in Table~\ref{tab:starbac_buildings}. The left column indicates the building's name and the right column describes the two-dimensional coordinates of the building's corners. There are five roles, which we describe next. We regard a permission as an \emph{action} executed on an \emph{object}.
\begin{table}[h!btp]
\centering
\begin{tabular}{|c|c|}
\hline
Name & Corners\\
\hline
\hline
Main building & $(1, 3), (1, 4), (4, 4), (4, 3)$\\
\hline
Library & $(1, 1), (1, 2), (2, 2), (2, 1)$\\
\hline
Station & $(8, 1), (8, 9), (9, 9), (9, 1)$\\
\hline
Laboratory & $(2, 6), (2, 8), (4, 8), (4, 6)$\\
\hline
Computer room & $(6, 6), (6, 7), (7, 7), (7, 6)$\\
\hline
\end{tabular}
\caption{}
\label{tab:starbac_buildings}
\end{table}

The first role assigns a permission to a user if all of the following hold:
\begin{itemize}
\item The user is at most 1 meter away from the computer room.
\item The object is in the computer room or in the laboratory.
\item The current day is an odd day of the month.
\item The current time is between 8AM and 5PM.
\end{itemize}

The second role assigns a permission to a user if all of the following hold:
\begin{itemize}
\item The user is outside the library.
\item The object is at most 1 meter away from the library.
\item Either
\begin{itemize}
\item the current day is before the 10th day of the month and the current time is between 2PM and 8PM or
\item the current day is after the 15th day of the month and the current time is between 8AM and 12PM.
\end{itemize}
\end{itemize}

The third role assigns a permission to a user if all of the following hold:
\begin{itemize}
\item The user is at most 3 meters away from the main building.
\item The object is at most 3 meters away from the main building.
\end{itemize}

The fourth role assigns a permission to a user if all of the following hold:
\begin{itemize}
\item The user is inside the library. 
\item The object is outside the library.
\item The current day is before the 15th day of the month.
\item The current time is between 12AM and 12PM.
\end{itemize}

The fifth role assigns a permission to a user if all of the following hold:
\begin{itemize}
\item The user is inside the main building, at most 1 meter away from the library, inside the laboratory, at most 2 meters away from the computer room, or inside the station.
\item The object satisfies the same spatial constraint.
\item The current day is before the 15th day of the month.
\item The current time is between 12AM and 12PM.
\end{itemize}